\documentclass[11pt]{amsart}
\usepackage{amssymb}
\usepackage{latexsym}
\usepackage[T1]{fontenc}
\usepackage{graphicx}
\usepackage{graphics}

 \newtheorem{thm}{Theorem}[section]
 \newtheorem{cor}[thm]{Corollary}
 \newtheorem{lem}[thm]{Lemma}
 \newtheorem{prop}[thm]{Proposition}
 \theoremstyle{definition}
 
 \theoremstyle{remark}

 \numberwithin{equation}{section}

\newcommand{\thmref}[1]{Theorem~\ref{#1}}

\newcommand{\corref}[1]{Corollary~\ref{#1}}

\newcommand{\defeq}{\stackrel{\rm{def}}{=}}
\def\squarebox#1{\hbox to #1{\hfill\vbox to #1{\vfill}}}
\newcommand{\stopthm}{\hfill\hfill\vbox{\hrule\hbox{\vrule\squarebox
                 {.667em}\vrule}\hrule}\smallskip}

\newcommand{\bra}[1]{\left\langle#1\right|}
\newcommand{\ket}[1]{\left|#1\right\rangle}
\newcommand{\moy}[1]{\langle #1 \rangle}
\newcommand{\lp}{\left(}
\newcommand{\rp}{\right)}

\newcommand{\lno}{\left\|}
\newcommand{\rno}{\right\|}
\newcommand{\fr}{\frac}
\newcommand{\dg}{\dagger}
\newcommand{\h}{\hbar}
\newcommand{\ti}[1]{\tilde {#1}}
\newcommand{\eps}{\varepsilon}
\newcommand{\trm}[1]{\textrm{#1}}
\def\hto0{\xrightarrow{h\to 0}}

\renewcommand{\d}{\partial}
\newcommand{\dbar}{\overline{\partial}}
\renewcommand{\i}{\textrm{i}}
\newcommand{\norm}[1]{\|#1\|}
\newcommand{\rc}{\sqrt}

\newcommand{\li}{{\mathcal L}}
\newcommand{\Sn}{{\mathcal S_n}}
\newcommand{\Mak}{M_h(a,\kappa)}

\newcommand{\cO}{{\mathcal O}}
\newcommand{\cU}{{\mathcal U}}
\newcommand{\cL}{\mathcal L}
\newcommand{\cI}{\mathcal I}
\newcommand{\Ohn}{{\mathcal O_{\hn}}}
\newcommand{\Ohhnz}{{\mathcal O_{\hn,z}}}
\newcommand{\Ohh}{{\mathcal O_{\hn}}}

\newcommand{\cS}{\mathcal S}
\newcommand{\cR}{\mathcal R}
\newcommand{\Hi}{{\mathcal H}}
\newcommand{\hn}{{\mathcal H_N}}
\newcommand{\ka}{\kappa}
\newcommand{\la}{\ell a}
\newcommand{\lc}{\ell c}
\newcommand{\lsn}{\ell \mathcal S_n}
\newcommand{\tsn}{\ti\cS_n}

\newcommand{\eia}{\operatorname{EI}_\infty(a)}
\newcommand{\esa}{\operatorname{ES}_\infty(a)}
\newcommand{\fc}{f_w}
\newcommand{\tfc}{\ti f_w}
\newcommand{\lan}{\ell a_n}

\newcommand{\ctc}{C^{\infty}(\t2)}

\newcommand{\C}{{\mathbb C}}
\newcommand{\N}{{\mathbb N}}

\newcommand{\R}{{\mathbb R}}
\newcommand{\Z}{{\mathbb Z}}
\newcommand{\SP}{{\mathbb S}}
\newcommand{\TT}{{\mathbb T}}

\newcommand{\E}{\mathbb E}
\def\t2{{\mathbb T}^2}
\newcommand{\bbbone}{{\mathchoice {1\mskip-4mu {\rm{l}}} {1\mskip-4mu {\rm{l}}}
{ 1\mskip-4.5mu {\rm{l}}} { 1\mskip-5mu {\rm{l}}}}}

\newcommand{\Supp}{\operatorname {Supp}}
\newcommand{\Ran}{\operatorname{Ran}}
\newcommand{\Tr}{\operatorname{Tr}}

\newcommand{\Spec}{\operatorname{Spec}}

\newcommand{\e}{\operatorname{e}}
\newcommand{\Op}{\operatorname{Op}}
\newcommand{\opn}{\operatorname{Op}_h}

\newcommand{\esss}{\operatorname{ess\ sup}}
\newcommand{\essi}{\operatorname{ess\ inf}}
\newcommand{\Id}{\operatorname{Id}}
\newcommand{\Fix}{\operatorname{Fix}}
\newcommand{\Diff}{\operatorname{Diff}}

\begin{document}

\title{Weyl laws for partially open  quantum maps}
\author{Emmanuel Schenck}

\address{
Institut de Physique Th\'eorique\\
CEA Saclay\\
91191 Gif-sur-Yvette\\
France}

\email{emmanuel.schenck@cea.fr}



\date{October 14, 2008}


\begin{abstract}
We study a toy model for ``partially open'' wave-mechanical system,
like for instance a dielectric micro-cavity, in the semiclassical
limit where ray dynamics is applicable. Our model is a quantized map
on the 2-dimensional torus, with an additional damping at each
time step, resulting in a subunitary propagator, or ``damped quantum
map''. We obtain analogues of Weyl's laws for such maps in the semiclassical limit, and draw some more precise estimates when the classical dynamic is chaotic.
\end{abstract}

\maketitle

\section{Introduction}
A quantum billiard $\Omega$ is a ``closed quantum system'', as it preserves  probability. Mathematically, this corresponds to the Laplace operator in $\Omega$ with
Dirichlet boundary conditions. In this case, the spectrum is discrete
and the associated eigenfunctions are bound states.
This system can be ``opened'' in various ways: among others, one possibility is to
consider the situation where the refractive index takes two
different values $n_{in/out}$ inside and outside the billiard. This  model can describe
certain types of two-dimensional optical microresonators: after some approximations,  the electromagnetic field satisfies the scalar Helmholtz equation $(\Delta+ k_{in/out}^2 )\Psi = 0$ inside and outside $\Omega$.  For transverse magnetic polarization 
of the electromagnetic field, $\Psi$ and $\nabla\Psi$ are continuous across
$\partial\Omega$, and  the
relation between $k_{in/out}$ and the energy $E$ is expressed by $k_{in/out}^2 = n^2_{in/out} E$. 
In
that case, the spectrum is purely absolutely continuous, and all
bound states are replaced by metastable states: they correspond to complex generalized eigenvalues, called \emph{resonances}, which are the poles of the meromorphic continuation of the resolvent from the upper to the lower half plane. These quantum resonances play a physically significant role, as their imaginary part govern the decay in time of the metastable states.

For such systems with a refractive index jump, the semiclassical (equivalently, the geometric optics) limit can be described as follows. A 
wavepacket travels along a single ray until in hits $\partial\Omega$, then it generally
{\it splits} between two rays, one reflected, the other refracted according to Snell's law.
If the cavity $\Omega$ is convex, the refracted ray will escape to infinity, and we may
concentrate to what happens inside: the wavepacket follows the same trajectory as in the
case of the closed cavity, but it is {\it damped} at each bounce by a reflection factor
depending on the incident angle. If we encode the classical dynamics inside the billiard
by the bounce map, then the effect of that reflection factor is to damp the wavepacket after
at each step (or bounce). This map  does \emph{not} preserve probability, it is a
``weighted symplectic map'', as studied in \cite{nz}.

The toy model we will study below has the same characteristics as
this bounce map, and  has been primarily introduced in \cite{kns} to mimick the resonance spectra of dielectric microcavities. It consists in a symplectic smooth map $\kappa$
acting on a compact phase space (the 2-dimensional torus $\t2$),
plus a damping function on that phase space $a\in \ctc$, with
$|a|\leq 1$. Throughout the paper we will not always assume 
precise dynamical property for the map $\ka$, althought the main
result take a very specific form when $\ka$ has the Anosov property. We will also assume that for each
time $n\geq 1$, the fixed points of $\kappa^n$ form a ``thin'' set
(the precise condition is given in \S \ref{s:dyn}).

The corresponding quantum system will
be contructed as follows: we first quantize $\kappa$ into a family of
$N\times N$ unitary propagators, where the quantum dimension
$N=(2\pi\hbar)^{-1}\defeq h^{-1}$ will be large, and the damping function $a$ is quantized into
an operator $\opn(a)$. All these quantities will be described in more detail in
\S\ref{s:setting}.
To have a damping effect at the quantum level,
we also need to assume that $\|\opn(a)\|\leq 1$ for
all $h\leq 1$.

If we denote $U_h(\kappa)$  the unitary  propagator obtained from a quantization of $\ka$, the damped quantum map then takes the form
\begin{equation}\label{Qmap}
\Mak\defeq\opn(a)\, U_h(\kappa)\,.
\end{equation}
Apart from the
ray-splitting situation described above, the above damped quantum
map is also relevant as a toy model for the {\em damped wave
equation} in a cavity or on a compact Riemannian manifold. Evolved
through the damped wave equation, a wavepacket follows a geodesic at
speed unity, and it is {\it continuously} damped along this
trajectory \cite{al,sjo}. The above damped quantum map is
a discrete time version of this type of evolution; it can be seen as
a ``stroboscopic'' or ``Poincar\'e'' map for such an evolution. To
compare the spectrum of our quantum maps with the complex modes 
$k_n=\omega_n-i\fr{\Gamma_n}{2}$ of a damped cavity, one should look
at the modes contained in an interval $|\omega_n-k|\leq \pi$ around
the frequency $k\approx h^{-1}$ : the distribution of the decay rates
$\Gamma_n : |\omega_n-k|\leq\pi$ is expected to  exhibit the same behavior as 
that of the decay rates $\{\gamma^{(h)}_n=-2\log |\lambda_n^{(h)}|\,, \, |\lambda_n^{(h)}|\in\Spec(\Mak)\}$, as we will see below.
Some 
of the theorems we  present here are analogues of theorems
relative to the spectrum for the damped wave equation, proved in
\cite{al,sjo}. Some of the latter theorems become trivial
in the present framework, while the proofs of some others simplifies
in the case of maps. Besides, the numerical diagonalization of
finite matrices is simpler than that of wave operators. Also, it is
easier to construct maps with pre-defined dynamical properties, than
manifolds with pre-defined properties of the geodesic flow.

\medskip

We now come to our results concerning the maps \eqref{Qmap}. Some of them -- Theorems \ref{th:ergo} and \ref{th:width} -- have already been presented without proof in \cite{ns}. 

In general, the matrix $\Mak$ is not normal, and may not be diagonalizable.
It is known that the spectrum of nonnormal matrices can be very sensitive to perturbations,
leading to the more robust notion of pseudospectrum \cite{et}.
We will show that, under the
condition of nonvanishing damping factor $|a|\geq a_{\min}>0$, the spectrum of $\Mak$
is still rather constrained in the semiclassical limit: it resembles
the spectrum of the unitary (undamped) map $U_h(\kappa)$.

\begin{thm}\label{th:strip}
Let $\Mak$ be the damped quantum  map described above, where $\ka$ is
a smooth, symplectic map on $\t2$ and the damping factor $a\in C^\infty (\t2)$ satisfies
$1\geq |a|>0$. For each value of $h=N^{-1},\, N\in\N$, we
denote by $\{\lambda_j^{(h)}\}_{j=1\ldots h^{-1}}$ the eigenvalues of
$\Mak$, counted with algebraic multiplicity. In the semiclassical
limit $h\to 0$, these eigenvalues are distributed as follows. 
Let us call
$$
a_n: x\mapsto  \prod_{i=1}^n|a\circ\ka^i(x)|^{\fr1n}\,,
$$
and using the Birkhoff ergodic theorem, define
$$
\eia = \essi \lim_{n\rightarrow\infty}a_n\,,\quad \esa = \esss\lim_{n\rightarrow\infty}a_n\,.
$$
Then the spectrum semiclassically concentrates near an annulus
delimited by the circles of radius $\eia$ and $\esa$: 
\begin{equation}\label{e:strip}
\forall \delta>0\,,\ \lim_{h\rightarrow 0}\,
h\, \#\Big\{1\leq j\leq h^{-1} : \eia -\delta\leq |\lambda_j^{(h)}| \leq \esa +\delta\Big\}= 1\,. 
\end{equation}
\end{thm}

\medskip

Suppose now that $\ka$ is ergodic with respect to the Lebesgue measure $\mu$ on $\t2$, and denote 
$$\moy a\defeq\exp\Big\{\int_{\t2} \log |a|d\mu\Big\}\quad\text{the geometric mean of $|a|$ on $\t2$}\,.$$ 
In this case, the spectrum concentrates near the circle of radius $\moy a$ and  the arguments of the eigenvalues become homogeneously
distributed over $\SP^1\simeq [0,1)$:
\begin{thm}\label{th:ergo}
If $\ka$ is ergodic with respect to the Lebesgue measure, 
\begin{align*}\label{e:arg}
&(i)\quad \forall\delta>0,\qquad \lim_{h\rightarrow 0}\,h\,\#\Big\{1\leq j\leq h^{-1} : \left| |\lambda_j^{(h)}| -\moy a \right|
\leq \delta \Big\}= 1\,,\\
&(ii)\quad  \forall f\in C^0(\SP^1),\qquad
\lim_{h\rightarrow0}\,h\sum_{j=1}^{h^{-1}} f \Big(\fr{\arg
\lambda_j^{(h)}}{2\pi} \Big)=\int_{\SP^1}f(t)\,dt\,.
\end{align*}
\end{thm}
Note that $(i)$ is an immediate consequence of \thmref{th:strip}, since  the Birkhoff ergodic theorem states that $\eia=\esa=\moy a$ if $\ka$ is ergodic with respect to  $\mu$. 
We remark that the ergodicity of $\ka$ ensures that the spectrum of
the \emph{unitary} quantum maps $U_h(\ka)$ become uniformly distributed as
$h\to 0$ \cite{bdb2,mok}. Actually, for this property to hold
one only needs a weaker assumption, contained in
Proposition \ref{p:Minkowski}.

If we now suppose that $\ka$ has the Anosov property (which implies ergodicity), we can estimate more precisely the behavior of the spectrum as it concentrates around the circle of radius $\moy a$ in the semiclassical limit.
\begin{thm}\label{th:width}
Suppose that $\ka$ is Anosov. Then, for any $\eps>0$,
\begin{equation}\label{e:width}
\lim_{h\rightarrow 0}\,h\,\#\left\{1\leq j\leq h^{-1} :
\left||\lambda_j^{(h)}|-\moy a\right|\leq \lp \fr{1}{\log
h^{-1}}\rp^{\fr{1}{2}-\eps}\right\}= 1\,.
\end{equation}
\end{thm}

This theorem is our main result, and is proven in $\S$\ref{s:anosov}. It relies principally on the knowledge of the rate of convergence of the function $a_n$ to $\moy a$ as $n\to\infty$, when $\ka$ is Anosov. 

In general, for chaotic maps the radial distribution of the spectrum around $\moy a$ 
does not shrink to 0 in the semiclassical limit, as numerical and analytical studies indicate \cite{al}.  Toward this direction, we estimate the  number of  ``large'' eigenvalues, i.e. the subset of the spectrum that stay at a finite distance $c>0$ from $\moy a$, as $h\to0$:  we show that this number is bounded by $h^{\nu-1}$, where $0<\nu<1$ can be seen as a ``fractal'' exponent and depend on both $a$ and $\ka$. One more time, we use for this purpose  information about the ``probability''  for the function $ a_n$ to take values away from  $ \moy a$, as $n$ becomes large. This involves \emph{large deviations properties} for $a_n$, which are  usually expressed in terms of a \emph{rate function} $I\geq0$ (see $\S$\ref{s:anosov}) depending on both $a$ and $\ka$.  For $d>0$, considering the  interval  $[\log \moy a +d, \infty[$, one has 
$$
\lim_{n\to\infty}\fr{1}{n}\log \mu\{x: \log a_n(x) \geq d +  \log \moy a\}=-I(d)\,.
$$
Our  last result then the takes the form:
\begin{thm}\label{th:large}
Suppose as before that $\ka$ is Anosov, and choose $c>0$. Define 
$$\Gamma=\log(\sup_x
\norm{D\ka|_x})\,,\quad \lc=\log (1+c/\moy a)\,,\quad a_-=\min_{\t2}|a|. $$
For any
constant $C>0$ and $\eps>0$ arbitrary small,  denote
$C^\pm=C\pm\eps$, and set $T_{a,\ka}=(2\Gamma - 12\log a_-)^{-1}$. 
Then, if $I$ denotes the rate function associated to $a$ and $\ka$, we have
\begin{equation*}
 h\,\#\{1\leq j\leq h^{-1}:|\lambda_j^{(h)}|\geq \moy a + c\}=\cO(h^{\nu_{a,\ka}(c)})\,,
\end{equation*}
where $\nu_{a,\ka}(c)=\fr{I(\lc^-)T_{a,\ka}^-}{1+I(\lc^-)T_{a,\ka}^-}$.
\end{thm}

At this time, we do not know if the upper bound in Theorem
\ref{th:large} is optimal. The proof of the preceding result involves the use of an  evolution time of order $n_\tau\approx \tau \log h^{-1}$, similar to an Ehrenfest time, up to which the quantum to classical correspondance -- also known as Egorov theorem -- is valid. The above bound is  optimal in the sense that a particular choice of $\tau=\tau_c$  makes minimal the bound we can obtain. It is given by 
$$
\tau_c\defeq\fr{T_{a,\ka}^-}{1+I(\lc^-)T_{a,\ka}^-}\,.
$$
It is remarkable that in our setting, a small Ehrenfest time does not gives any relevant bound, but a large Ehrenfest time does not provide an optimal bound either, because the remainder terms in the Egorov theorem become too large. 

We also find interesting to note that in the context of the damped wave equation, a result equivalent to Theorems 
\ref{th:strip} and \ref{th:ergo} was obtained by Sj\"ostrand under the assumption
that the geodesic flow is ergodic \cite{sjo}.  But to our knowledge, no
results similar to \thmref{th:width} are known in this framework. 
 Concerning Theorem \ref{th:large}, a   comparable result has been announced  very
recently in the case of the damped wave equation on manifolds of negative curvature \cite{ana},  but working with flows on manifolds  add  technical complications compared to
our framework.  Although, it could be appealing to compare the nature of the upper bounds obtained in these two formalisms, and in both cases, it is still an open question to know if any lower bound could be determined for the number of eigenvalues larger than $\moy a+c$, in the semiclassical limit. We also note  that this ``fractal Weyl law'' is different from the fractal law for resonances presented in \cite{sz}, although the two systems share some similarities.

In the whole paper, we discuss quantum maps on the 2-torus $\t2$.
The generalization to the $2n-$dimensional torus, or any
``reasonable'' compact phase space does not present any new
difficulty, provided a quantization can be constructed on it, in the
spirit of \cite{mok}. 

In section \ref{s:setting}, we introduce the general setting of quantum mechanics and pseudodifferential calculus on the torus. Theorems \ref{th:strip} and \ref{th:ergo}, together with some  intermediate results are proved in section \ref{s:density}, while the  Anosov case is treated in section \ref{s:anosov}. 
In section \ref{s:numeric}, we present numerical calculations of the
spectrum of such maps to illustrate theorems \ref{th:strip}, \ref{th:ergo} and  \ref{th:width}. The observable $a$ is chosen somewhat
arbitrarily (with $|a|>0$), while $\ka$ is a well-known perturbed cat map. 

\section{Quantum mechanics on the torus $\t2$}\label{s:setting}
We briefly recall the setting of quantum mechanics on the 2-torus.
We refer to the literature for a more detailed presentation
\cite{hb,bdb1,deg}.

\subsection{The quantum torus}
When the classical phase space is the torus
$\t2\defeq\{x=(q,p)\in(\R/\Z)^2\}$, one can define a corresponding
quantum space by imposing periodicity conditions in position and
momentum on wave functions. When Planck's constant takes the
discrete values $\hbar=(2\pi N)^{-1}$, $N\in\N$, these conditions
yield a subspace of finite dimension $N$, which we will denote by
$\hn$. This space can be equipped with a ``natural'' hermitian
scalar product.

We begin by fixing the notations for the $\hbar$-Fourier transform
on $\R$, which maps position to momentum. Let $\cS$ denote the
Schwartz space of functions and $\cS'$ its dual, i.e. the space of
tempered distributions. The $\hbar$-Fourier transform of any
$\psi\in \cS' (\R)$ is defined as
$$
F_\h\psi(p)=\fr{1}{\sqrt{2\pi\h}}\int\psi(q)e^{-\fr{i}{\h}qp}dq\,.
$$
A wave function on the torus is a distribution periodic in both
position and momentum: 
\begin{equation*}
\psi(q+1)=e^{2i\pi\theta_2}\psi(q),\ \
F_\h\psi(p+1)=e^{2i\pi\theta_1}F_\h\psi(p) \,.
\end{equation*}
Such distributions
can be nontrivial iff $\h=(2\pi N)^{-1}$ for some $N\in\N$, in which
case they form a subspace $\hn$ of dimension $N$. For simplicity we
will take here $\theta_1=\theta_2=0$. Then this space admits a
``position'' basis $\{\ket{e_j} : j\in\Z/N\Z\}$, where
\begin{equation}\label{basis}
 \bra q e_j\rangle =\fr{1}{\sqrt
N}\sum_{\nu\in\Z}\delta(q-\nu-j/N)\,. 
\end{equation}
A ``natural'' hermitian
product on $\hn$ makes this basis orthonormal: 
\begin{equation}\label{inner}
\bra {e_j}e_k\rangle=\delta_{jk},\quad j,k\in \Z/N\Z\,,
\end{equation}
and we will denote by $\|\cdot\|$ the corresponding norm.

Let us now describe the quantization of observables on the torus. We
start from pseudodifferential operators on $L^2(\R)$ \cite{gs,ds}. To
any $f\in\cS(T^*\R)$ is associated its $\h-$Weyl quantization, that
is the operator $f^w_\h$ acting on $\psi\in \cS(\R)$ as:
\begin{equation}\label{WeylQuantiz}
f^w_\h\psi (q)\defeq \fr{1}{2\pi\h}\int
f(\fr{q+r}{2},p)e^{\fr{i}{\h}(q-r)p}\psi(r)\,dr\,dp\,.
\end{equation}
This defines a continuous mapping from $\cS$ to $\cS$, hence from
$\cS'$ to $\cS'$ by duality. Furthermore, it can be shown that the
mapping $f\mapsto f^w_\h$ can be extended to any $f\in
C^\infty_b(T^*\R)$, the space of smooth functions with bounded
derivatives, and the Calder\'on-Vaillancourt theorem shows that $f^w_\h$ is
also continuous on $L^2(\R)$.

A complex valued observable on the torus $f\in C^\infty (\t2)$ can be identified
with a biperiodic function on $\R^2$ (for all $q,p$,
$f(q+1,p)=f(q,p+1)=f(q,p)$). When $\h=1/2\pi N$, one can check that
the operator $f^w_\h$ maps the subspace $\hn\subset \cS'(\R)$ to
itself.  In the
following, we will always adopt the notation $h\defeq 2\pi\h$, so that on the torus we have $h=N^{-1}$. This  number $h$  will 
play the role of a small parameter and remind us the standard
$\h$-pseudodifferential calculus in $T^*\R$. We will then write $\opn(f)$ 
for the \emph{restriction} of
$f^w_\h$ on $\hn$, which will be the quantization of $f$ on the
torus. It is a $N\times N$ matrix in the basis \eqref{basis}.

The operator $\opn(f)$ inherits some properties from $f^w_\h$. We
will list the ones which will be useful to us. $\opn(f)^\dg =
\opn(f^*)$, so if $f$ takes real values, $\opn(f)$ is self-adjoint. The function  $f\equiv f_h$ may also depend on $h$, 
and to keep on the torus the main features of the standard pseudodifferential calculus in $T^*\R$,  these functions -- or symbols -- must belong to particular classes. On the torus, these
different classes are defined exactly as for symbols in
$C_b^\infty(T^*\R)$. For a sequence of functions
$(f_\h)_{\h\in ]0,1]},\ f_{\h}\equiv f(x,\h)\in C_b^\infty(T^*\R)\times]0,1]$, we will say that $f_\h\in S_\delta^m(1)$, with
$\delta\in]0,\fr{1}{2}]$, $m\in\R$ if $\h^{m} f_\h$ is uniformly bounded with respect to $\h$ and for any
multi-index $\alpha=(n_1,n_2)\in \N^2$ of length $|\alpha|=n_1 + n_2$, we have :
\begin{equation*}
\|\d^\alpha f_\h\|_{C^0} \leq C_\alpha \h^{-m -|\alpha|\delta}\,.
\end{equation*}
In the latter equation,  $\|\cdot\|_{C^0}$ denotes the sup-norm on $\t2$ and $\d^\alpha$ stands for $\fr{\d^{n_1+n_2}}{\d q^{n_1}\d p^{n_2}}$. 
On the torus, one simply has $2\pi\h=h=N^{-1}$ for some $N\in\N$, and $C^\infty_b(T^*\R)$ is replaced by $\ctc$. Let us denote $S_\delta^m$ these symbol classes.  We have the following inequality of norms, useful to carry properties of $\h-$pseudodifferential operators on $\R$ to the torus \cite{bdb1}: 
\begin{equation}\label{norm}
\forall f\in \ctc\, ,\quad \| \opn(f)\| \leq \|
f^w_\h \|_{L^2\to L^2}\,. 
\end{equation}
Note that this property remains valid if $f\equiv f_\h $ depends on $\h$, with $f_\h\in S_\delta^0(1)$. 
The $L^2$ continuity states that if $f_\h\in S_\delta^0(1)$, then 
$f_\h^w$
is a bounded operator (with $\h-$uniform bound) from $L^2(\R)$ to
$L^2(\R)$.  Since $\opn(f)$ is the restriction of $f^w_\h$ on $\hn$,
Eq. \eqref{norm} implies the existence of a constant $C$ independent
of $h$ such that for $f\in S_\delta^0$ and $h\in]0,1]$, 
\begin{equation*}
\|\opn(f)\|\leq C\,. 
\end{equation*}

The symbol calculus on the Weyl operators $f^w_\h$ easily extends to
their restrictions on $\hn$. We have, for  symbols $f$ and $g$ in $S_\delta^0$ the composition rule:
$$
\opn(f)\opn(g)=\opn(f\sharp_h g)\,,
$$
where $f\sharp_h g $ is defined as usual: for $X=(x,\xi)\in \R^2$ we have
\begin{equation}\label{e:compo}
(f\sharp_h g) (X)=\left( \fr{2}{h}\right)^2\int_{\R^{4}} f(X+Z)g(X+Y)\e^{\fr{4\i\pi}{h}\sigma (Y,Z)}dYdZ
\end{equation}
where $\sigma$ denotes the usual symplectic form.
Another useful expression for $f\sharp_h g$ is given by 
\begin{equation}\label{e:compo2}
(f\sharp_h g)(X)=\e^{\fr{\i h}{2\pi }\sigma(D_X,D_Y)}(f(X)g(Y))|_{X=Y}
\end{equation}
where $D_X=(D_x,D_\xi)=(\fr{1}{\i}\d_x,\fr{1}{\i}\d_\xi)$. This representation is particulary useful for $h-$expansions of $f\sharp_h g$. 
If $f\in S_\delta^m$ and 
$g\in S_\delta^n$, then $f\sharp_h g \in
S_\delta^{m+n}$, and at
first order we have :
\begin{equation}\label{e:symbcalc}
\|\opn(f)\opn(g)-\opn(f\,g)\|\leq C_{f,g} h^{1-2\delta-m-n}\,. 
\end{equation}
The Weyl operators on $\hn$, $\hat T_{m,n}\defeq
\opn(e_{mn})$, with
$$e_{mn}(q,p)\defeq e^{2i\pi(m q-n
p)}\,, m,n\in\Z\,,$$
 allow us to represent $\opn(f)$: 
\begin{equation}\label{e:fourier}
\opn(f)=\sum_{m,n\in\Z} f_{m,n}\,\hat
T_{m,n}\,, \textrm{ where } f_{m,n}=\int_{\t2} f\,
\overline{e_{mn}}\,d\mu\,. 
\end{equation}
From the trace identities
\begin{equation}
\Tr\hat T_{\mu,\nu}=
\begin{cases}
(-1)^{h\mu\nu}\, h^{-1} & \text{if}\ \mu,\nu=0\mod h^{-1}\\
0 & \text{otherwise},
\end{cases}
\end{equation}
one easily shows that for any $f\in S_\delta^0$ we have
\begin{equation}\label{e:trace} 
h\Tr(\opn(f))=f_{0,0} +\cO(h^{\infty}) =
\int_{\t2} f\,d\mu +\cO(h^\infty)\,.
\end{equation}
\medskip

Let $a\in S_0^0$. We will write $a_-=\min_{\t2} |a|$,
and $a_+=\max_{\t2} |a|$. The next proposition adapts the 
sharp G\aa rding inequality to the torus setting.
\begin{prop}\label{p:gard} Let $a\in S_\delta^0$ be a real, positive symbol, with $\Ran a=[a_-,a_+]$.
There exist a constant $C>0$ such that, for small enough $h$ and
any normalized state $u\in\hn$:
$$
a_- - C h^{1-2\delta} \leq \langle u, \opn(a) u\rangle \leq  a_+ +
C h^{1-2\delta}\,.
$$
\end{prop}
\begin{proof}
We first sketch the proof of the sharp G\aa rding inequality in the
case of pseudodifferential operators on $T^*\R$ with real symbol $a\in
S_\delta^0(1)$.  For the lower bound, we suppose without loss of
generality that $a_-=0$.

For $X=(x,\xi)\in\R^2$, consider
$\Gamma(X)=\Gamma(x,\xi)=2\e^{-\fr{x^2+\xi^2}{\h}}$. Using Eq. \eqref{e:compo} for symbols on $T^*\R$, 
a straightforward calculation involving Gaussian integrals shows that $\Gamma\sharp_h\Gamma=\Gamma$,
hence $\Gamma^w_\h$ is an orthogonal projector, and thus  a
positive  operator. Define the symbol
$$
 a \star\Gamma (X)\defeq \fr{1}{2\pi\h}\int_{\R^{2}} a(X+Y) 2\e^{- \fr{Y^2}{\h}}dY\,. 
$$
To connect $(a\star\Gamma)_\h^w$ with $a_\h^w$, 
we  make use of the Taylor formula at point $X$ in the above definition. Because of the
parity of $\Gamma$ and the fact that $\int_{\R^2}\Gamma(X)dX=2\pi\h$, we have
\begin{eqnarray*}
a\star\Gamma(X)&=&a(X)+\fr{2}{2\pi\h}\iint(1-\theta) a''(X+\theta
Y)Y^2\,  \e^{-\fr{Y^2}{\h}}dYd\theta\\ 
 &=&a(X)+\underbrace{\fr{1}{\pi}\iint(1-\theta) \h\, a''(X+\theta
Z\rc\h)Z^2\,  \e^{-Z^2}dZd\theta}_{r(X)}\,.
\end{eqnarray*}
To evaluate $\|r_\h ^w \|_{L^2\to L^2}$, we rescale the variable
$X\mapsto \ti X=X/\rc\h$, and call $\ti r( X)=r(\rc\h X)$. If
$u\in L^2$,  we denote $\ti u(\ti x)=\h^{\fr{1}{4}}u(x)$. This transformation
is unitary : $\norm {\ti u}_{L^2}=\norm{u}_{L^2}$. Now,  a simple change of variables using \eqref{WeylQuantiz} shows that  $\|r_\h^w u\|_{L^2}=\|\ti r_1^w \ti u\|_{L^2}$ where $\ti r_1^w$ denote the $\h=1$ quantization of the symbol $\ti r$. Because of the term $\h\,a''=\cO(h^{1-2\delta})$  appearing in the definition of $r$, for any multi-index $|\alpha |$ we have
$$
\d^{\alpha}_{\ti X} \ti r(\ti X)\lesssim
\h^{1+\fr{|\alpha|}{2}}\h^{-\delta(|\alpha|+2)}=\h^{1-2\delta}\h^{|\alpha|(\fr12-\delta)}\,. 
$$
Since the $L^2$ continuity theorem  applied to $\ti r^w_1$ yields to a bound  that involves a finite number of derivatives of $\ti r$, we get $\norm{\ti r_1^w}_{L^2\to L^2}=\cO(\h^{1-2\delta})$.  Here and below, by $f\lesssim g$ we will  mean that $|f|\leq C |g|$ for some $C\geq0$. Hence,
\begin{equation*}
\norm{r^w_\h  u}_{L^2}=\norm{\ti r_1^w \ti u}_{L^2}\leq \norm{\ti r_1^w}_{L^2\to L^2}\norm{u}_{L^2} \lesssim
\h^{1-2\delta}\norm{u}_{L^2}\,,
\end{equation*}
 and we conclude by
$
\|r_\h^w\|_{L^2\to L^2}=\cO(\h^{1-2\delta})\,. 
$

 Now, from the definition of $a\star\Gamma$, we also have  
$$(a\star\Gamma)_\h^w=\fr{1}{2\pi\h}\int a(Y)(\Gamma(\cdot -Y))_\h^w dY\,.$$ But as we noticed above, $ (\Gamma(\cdot-Y))_\h^w>0$, and then  $(a\star\Gamma)_\h^w$ is positive definite. Using the fact that $a\star\Gamma=a+r$ and $
\|r_\h^w\|_{L^2\to L^2}=\cO(\h^{1-2\delta}) 
$, this implies the existence of a constant $c>0$ such that
\begin{equation}\label{e:gard}
\langle u, a_\h^w u\rangle \geq -c\h^{1-2\delta}\,. 
\end{equation}
The upper bound is obtained similarly, assuming $a_+=0$ and
considering the symbol $-a\geq 0$.

It is now a straightforward calculation to transpose these properties on the torus by making use of Eq. \eqref{norm}.
\end{proof}

\subsection{Quantum dynamics}\label{s:dyn}

The classical dynamics will simply be given by a smooth symplectic
diffeomorphism $\ka:\t2\to\t2$. Since we are mainly interested in the case of chaotic dynamics, we will sometimes make  the hypothesis that $\ka$ is Anosov, hence ergodic with respect to the Lebesgue measure $\mu$. From this ergodicity we draw
the following consequence on the set of periodic points. We recall
that a set has {\em Minkowski content zero} if, for any $\eps>0$, it
can be covered by equiradial Euclidean balls of total measure less
than $\eps$.
\begin{prop}\label{p:Minkowski}
Assume the diffeomorphism $\kappa$ is ergodic w.r.to the Lebesgue
measure. Then, for any $n\geq 1$, the fixed points of $\ka^n$ form a
set of Minkowski content zero.
\end{prop}
\begin{proof}
Let us first check that, for any $n\neq 0$, the set of $n$-periodic
points $\Fix(\kappa^n)$ has Lebesgue measure zero. Indeed, this set
is $\kappa$-invariant, so by ergodicity it has measure $0$ or $1$.
In the latter case, the map $\kappa$ would be $n$-periodic on a set
of full measure, and therefore not ergodic.

Let us now fix $n\neq 0$. Since $\kappa^n$ is continuous, for any
$\epsilon\geq 0$ the set
$$
F_\epsilon\defeq \{x\in\t2,\;{\rm
dist}(x,\kappa^n(x))\leq\epsilon\}\qquad\text{is closed in $\t2$.}
$$
Since $F_\epsilon\subset F_{\epsilon'}$ if $\epsilon\leq\epsilon'$,
for any Borel measure $\nu$ on $\t2$ we have
$$
\nu(F_0)=\lim_{\epsilon\to 0}\nu(F_\epsilon)\,.
$$
Since $F_0=\Fix(\kappa^n)$ has zero Lebesgue measure, it is of
Minkowski content zero.
\end{proof}

We will not study in detail the possible quantization recipes of the
symplectic map $\ka$ (see e.g. \cite{kmr,deg,zel} for discussions on
this question), but assume that some quantization can be constructed. In dimension $d=1$, the map
$\ka$ can be decomposed into the product of three maps $L,t_\mathbf{v}$ and
$\phi_1$ where $L\in SL(2,\Z)$ is a linear automorphism of the torus,
$t_\mathbf{v}$ is the translation of vector $\mathbf{v}$ and $\phi_1$ is a time 1
hamiltonian flow \cite{kmr}. One quantizes the map $\ka=L\circ t_\mathbf{v}\circ
\phi_1 $ by quantizing separately  $L,t_\mathbf{v}$ and $\phi_1$ into
$U_h(L),U_h(t_\mathbf{v})$ and $U_h(\phi_1)$ and setting $U_h(\ka)=U_h(L)\,U_h(t_\mathbf{v})\, U_h(\phi_1)$. We are mainly interested in the Egorov
property, or ``quantum to classical correspondence principle'' of such maps, which is
expressed for $f\in S_\delta^0$  by
\begin{equation*}
\norm{U_h(\ka)^{-1}\opn(f)U_h(\ka)-\opn(f\circ\ka)}=o_h(1)\,,\quad\trm{where}\quad o_h(1)\hto0 0\,.
\end{equation*}
The
following lemma makes  the defect term $o_h(1)$ in the preceding equation more precise. 
\begin{lem}
Let $f\in S_\delta^0$. There is a constant $C_{f,\ka}$ such that
\begin{equation}\label{e:Egorov}
\norm{U_h(\ka)^{-1}\opn(f)U_h(\ka)-\opn(f\circ\ka)}\leq C_{f,\ka} h^{1-2\delta}\,. 
\end{equation}
\end{lem}

\begin{proof} 
Since it is well known that for linear maps $L$, one has
$$U_h(L)^{-1}\opn(f)U_h(L)=\opn(f\circ L)\,,$$
we will consider only
the quantizations of $t_\mathbf{v}$ and $\phi_1$. The map $t_\mathbf{v}$ is quantized
by a quantum translation operator of vector $\mathbf{v}_h$, which is at
distance  $|\mathbf{v}-\mathbf{v}_h|=\cO(h)$:
$$
U_h(t_\mathbf{v})^{-1}\opn(f)U_h(t_\mathbf{v})=\opn(f\circ t_{\mathbf{v}_h})\,.
$$
Consequently, we need to estimate 
$\norm{\opn(f\circ t_\mathbf{v}-f\circ t_{\mathbf{v}_h})}$. For this purpose, we first evaluate $\|f\circ t_\mathbf{v}-f\circ t_{\mathbf{v}_h}\|_{C^0}$. Denote $\mathbf{v}=(v^q,v^p)$ and $\mathbf{v}_h=(v_h^q,v_h^p)$. A Taylor expansion shows that 
$$
\|f(q+v^q_h+(v^q - v^q_h),p+v^p_h+(v^p -v^p_h))-f(q+v^q_h,p+v^p_h)\|_{C^0}=\cO(h^{1-\delta})
$$
and then,
$$
\|f\circ t_\mathbf{v}-f\circ t_{\mathbf{v}_h}\|_{C^0}=\cO(h^{1-\delta})\,.
$$
Hence, $h^{\delta-1}(f\circ t_\mathbf{v}-f\circ t_{\mathbf{v}_h} )\in S_\delta ^0$.
Using Proposition \ref{p:gard}, we conclude by 
\begin{eqnarray*}
\norm{\opn(f\circ t_\mathbf{v}-f\circ t_{\mathbf{v}_h})} =\cO(h^{1-\delta})\,.
\end{eqnarray*}

Let us denote by $H\in S_0^0 $ the Hamiltonian, $X_H$ the associated Hamiltonian vector field, and $\phi_{t}=\exp\lp t X_H \rp$ the classical Hamiltonian flow at time $t$. For simplicity, we denote the quantum propagator at time $t$  by $\cU^t\defeq \exp(-\fr{\i t}{h}\Op_h (H))$, and write $\phi_t^* f\defeq f\circ\phi_t$. In particular, note that $\cU^1=U_h(\phi_1)$. From the equations: 
$$
\begin{cases}
\begin{array}{l}
\fr{d}{ds}\phi_{t+s}^*f|_{s=0}=\{H,\phi_t^*f\}\\
\fr{d}{ds}\cU^{t+s}\Op_h (\phi_{t+s}^*f)\cU^{-(t+s)}|_{s=0} =\cU^t \Diff_t(H,f)\cU^{-t}
\end{array}
\end{cases}
$$ 
with $ \Diff_t(H,f) = \opn(\{H,\phi_t^* f\}) -\fr{\i}{h}[\opn(H),\opn (\phi_t^*f)]$, 
we get 
$$
\Op_h(\phi_t^* f)=
\cU^{-t} \Op_h (f) \cU^{t} +\int_0^t \cU^{s-t}\Diff_s (H,f) \cU^{t-s} ds\,.
$$
A straightforward application of \eqref{e:compo2} yields to 
$$
\fr{2\i\pi}{h}[\opn(H),\opn(\phi_t^*f)]=\opn(\{H,\phi_t^*f\})+\fr{\i}{h}\Ohn(h^{2-2\delta})\,,
$$
where $\Ohn(q)$ denotes an operator in $\hn$ whose norm is of order $q$.  
If we  use the unitarity of $\cU^t$, we  thus obtain :
$$
\|\cU^{-1}\Op_h(f)\cU^1 - \opn(\phi_1^* f)\|=\cO(h^{1-2\delta})\,.
$$
Adding  all these estimates, we end up with
\begin{equation*}
\norm{U_h(\ka)^{-1}\opn(f)U_h(\ka)-\opn(f\circ\ka)}\leq C_{f,\ka} h^{1-2\delta}\,.
\end{equation*}
\end{proof}

Taking into account the damping, our damped quantum map is given by
the matrix $\Mak$ in \eqref{Qmap}. The damping factor $a\in S_0^0$ is
chosen such that, for $h$ small enough, $\|\opn(a)\|\leq 1$ and
$a_-=\min_{\t2} |a|>0$. From Proposition \ref{p:gard}, this implies that
$\opn(a)$, and thus $\Mak$, are invertible, with inverses uniformly
bounded with respect to $h$: 
\begin{equation}\label{e:bounds} 
\|\Mak \|= a_+ +
\cO(h),\qquad \|\Mak^{-1} \|= a_-^{-1} + \cO(h)\,.
\end{equation}
 As
explained in the introduction, $\Mak$ is not a normal operator, and
it  may not be diagonalizable. Nevertheless, we may write its
spectrum as
$$
\Spec(\Mak)=\{\lambda^{(h)}_1,\lambda^{(h)}_2,\ldots,\lambda^{(h)}_{h^{-1}}\}\,,
$$
where each eigenvalue is counted according to its algebraic
multiplicity, and eigenvalues are ordered by decreasing modulus (in
the following we will sometimes omit the $^{(h)}$ supersripts). The
bounds \eqref{e:bounds} trivially imply 
\begin{equation}\label{e:strict} 
a_+ +
C h\geq
|\lambda^{(h)}_1|\geq|\lambda^{(h)}_2|\geq...\geq|\lambda^{(h)}_{h^{-1}}|\geq
a_- - Ch\,,
\end{equation}
 for some constant $C>0$. Since we assumed $a_->0$,
the spectrum is localized in an annulus for $h$ small enough. The
above bounds are similar with the ones obtained for the damped wave
equation \cite[Eq.(2-2)]{al}.

For later use, let us now recall the Weyl inequalities \cite{kon},
which relate the eigenvalues of an operator $A$ to its singular
values (that is, the eigenvalues of $\sqrt{A^\dg A}$):
\begin{prop}[Weyl's inequalities]\label{p:weylineq}
Let $\Hi$ be a Hilbert space and  $A\in\li(\Hi)$ a compact operator.
Denote respectively $\alpha_1,\alpha_2,\ldots$ (respectively
$\beta_1,\beta_2,\ldots$) its eigenvalues (resp. singular values)
ordered by decreasing moduli and counted with algebraic
multiplicities. Then, $\forall k\leq \dim\Hi$,  we have :
\begin{equation}\label{WeylIneq}
 \prod_{i=1}^k|\alpha_i|\leq\prod_{i=1}^k
\beta_i\ ,\qquad \sum_{i=1}^k|\alpha_i|\leq \sum_{i=1}^k \beta_i
 \end{equation}
\end{prop}
We immediately deduce from this the following
\begin{cor}\label{c:Weyl-ineq} Fix $n\geq 1$.
Let  $\lambda_1,\lambda_2,\ldots $ be the eigenvalues of $A$ and
$s_1^{(n)},s_2^{(n)},\ldots$ the eigenvalues of $A_n= \sqrt[2n]{A^{\dg
    n}A^n}$, ordered as above.
Then for any $k\leq\dim\Hi$:
$$
\sum_{i=1}^k\log|\lambda_i|\leq\sum_{i=1}^k\log s_i^{(n)}\,, \qquad \prod_{i=1}^k|\lambda_i|\leq\prod_{i=1}^k
s_i^{(n)}\,.
$$
\end{cor}

\subsection{A first-order functional calculus on $\li(\hn)$}

In \S\ref{s:density} we will need to analyze the operators
$$
\{\sqrt[2n]{\Mak^{\dg n}\;\Mak^n},\ \ n\in\N\}\,.
$$
We will show that they are (self-adjoint) pseudodifferential
operators (i.e. ``quantum observables'') on $\hn$, and then draw
some estimates on their spectra in the semiclassical limit via
counting functions and trace methods. We will use for this a
functional calculus for operators in $\li(\Hi_N)$, obtained from a
Cauchy formula, via the method of almost analytic extensions.

 Let $a\in S_\delta^0$ be a real symbol. In order to localize the spectrum of
$\opn(a)$ over a set  depending on $h$, we
will make use of compactly supported functions $\fc\in C_0^\infty(\R)$
which can have variations of order 1 over distances of order $w(h)$, for some continuous
function $w>0$, satisfying  $w(h)\hto0 0$. 
 
The functions $\fc$ have  derivatives growing as $h\to 0 $:  for any $m\in\N$, we will assume that 
\begin{equation}\label{h:countfunct}
\|\d^m \fc\|_{C^0}\leq C_m w(h)^{-m}\,.
\end{equation}
We will call such functions 
$w(h)-$ admissible. 
Our main goal in this section
consist in defining the operators $\fc(\opn(a))$ and characterize
them as pseudodifferential operators on the torus.

We first construct an almost analytic extension
$\tfc\in C_0^{\infty}(\C)$ satisfying
\begin{equation}\label{tf1}
 \|\dbar \tfc\|_{C^0}\leq C_m |\Im z|^m w(h)^{-m-2}\ ,\ \forall m\geq0
\end{equation}
\begin{equation}\label{tf2}
\tfc|_\R=\fc\,,
\end{equation}
where $\dbar$ stands for $\frac{\d}{\d\bar
z}=\frac{1}{2}(\frac{\d}{\d x}+\i\frac{\d}{\d y})$. For this purpose, we follow closely \cite{ds}  : one put $\chi\in C_0^\infty(\R)$ be equal to 1 in a neighborhood of
0, $\psi \in C_0^\infty(\R)$ equal to 1 in a  neighborhood of
$\Supp(f_w)$ and define
\begin{equation*}
 \tfc(x+\i y)=\frac{\psi(x)\chi(y)}{2\pi}\int \e^{\i(x+\i y)\xi}\chi(y\xi)\hat f_w (\xi)d\xi\,.
\end{equation*}
We will now check that $\tfc$ is a function that satisfies Eqs.
\eqref{tf1} and \eqref{tf2}. Notice first that $\tfc$ is compactly supported in $\C$, and that Eq. \eqref{tf2} follows clearly from
the Fourier inversion formula. To derive \eqref{tf1}, remark  that
\begin{eqnarray*}
\dbar \tfc &=&\fr{\i\psi(x)\chi(y)}{4\pi}y^m\int \e^{\i(x+\i
  y)\xi}\fr{\chi'(y\xi)}{(y\xi)^m}\xi^{m+1}\hat f_w(\xi)d\xi\\
& &+\,\fr{\psi'(x)\chi(y)+\i\psi(x)\chi'(y)}{4\pi}\int \e^{\i(x+\i y-\ti x)\xi}\chi(y\xi)f_w(\ti
x)d\ti x d\xi \\
&=&\trm{I} + \trm{II}\,.
\end{eqnarray*}
Note that  if $y=0$, $\trm{I}=\trm{II}=0$ because of the properties of $\chi$ and $\psi$,  so \eqref{tf1} is satisfied with $C_m =0$. We now suppose $y\neq 0$. Since $f_w$ is compactly supported, we can integrate by parts and
using \eqref{h:countfunct}, we obtain, setting $t=y\xi$ :
\begin{equation*}
\left|\int  \fr{\chi'(y\xi)}{(y\xi)^m} \xi^{m+2}\hat f_w (\xi)\fr{d\xi}{\xi}\right|\leq
\int \left|  \fr{\chi'(t)}{t^{m+1}}
(\d^{m+2}_x\fc(x))\right|dt dx \leq 
C_m  w(h)^{-m-2}\,.
\end{equation*}
In the last inequality, we used the fact that $f_w$ and $t\mapsto t^{-m-1}\chi'(t)$ are compactly supported. Hence, this shows that  $|\trm{I}|\leq C_m |y|^m w(h)^{-m-2}$. To treat the term $\trm{II}$, denote 
$$
F(x,y)=\fr{\psi'(x)\chi(y)+\i\psi(x)\chi'(y)}{4\pi}\ \ \trm{and}\ \ G(x,\ti x,y)=(\i+D_{\ti x})^2 D_{\ti
  x}^m(\fr{f_w(\ti x)}{x-\ti x +\i y})\,.
$$ 
Since $y\neq 0$, we can
rewrite II to get
$$
\trm{II}=\i F(x,y)\int y^m \e^{\i(x-\ti x +\i
  y)\xi}\fr{\chi'(y\xi)y}{(\xi+\i)^2(y\xi)^m}G(x,\ti x, y) d\ti x d\xi\,.
$$
As above, we set $t=y\xi$. This gives
\begin{eqnarray*}
\int \left|  y^m \fr{\chi'(y\xi)y}{(\xi+\i)^2(y\xi)^m}G(x,\ti x, y) 
\right|d\ti x d\xi&\leq& \int \left| y^{m+2} \fr{\chi'(t)}{(t+\i y)^2t^m}G(x,\ti x, y)  \right|dtd\ti x\\
&\leq& \int \left| y^{m+2} \fr{\chi'(t)}{t^{m+2}} G(x,\ti x,y) \right| dtd\ti x\,.
\end{eqnarray*}
Let us distinguish two cases.
\begin{itemize}
\item  If $x\notin\Supp\psi$, then $\ti x\in \Supp f_w\Rightarrow \ x-\ti x\neq 0$, 
from which we deduce that $\int |G(x,\ti x,y )|d\ti x\leq c_m w(h)^{-m-2}$.
\item If $x\in \Supp\psi$, then $F(x,y)=\i(4\pi)^{-1}\chi'(y)$. 
If $\chi'(y)=0$, we get $\trm{II}=0$.  Otherwise, we have $b_1\leq |y|\leq b_2 $ for some fixed constants $b_1,b_2>0$ depending on $\chi$. In this case, we get again  $\int |G(x,\ti x,y )|d\ti x\leq c_m w(h)^{-m-2}$. 
\end{itemize}
Grouping the results, we see that $|\trm{II}|\leq C_m |y|^{m+2} w(h)^{-m-2}\leq \ti C_{m}|y|^{m}w(h)^{-m-2}$, and since  $|\trm{I}|\leq C_m |y|^m w(h)^{-m-2}$, it follows that  Eq. \eqref{tf1} holds.

Considering a function $\fc$ as above, we now  characterize
$\fc(\opn(a))$ as a pseudodifferential operator on the torus. We begin by two
lemmas concerning resolvent estimates. For $z\notin\Spec\opn(a)$, we
denote $\cR_z(a)=(\opn(z-a))^{-1}$ the resolvent of $\opn(a)$ at point
$z$.
\begin{lem}\label{l:rsv-symb} Let $a\in S_\delta^0$ be a real symbol, and
$\Omega\subset\C$ a bounded domain such that $\Supp
(\tfc)\subset\Omega$. Take $z\in\Omega$ and suppose that $|\Im
z|\geq h^{\eps}$ for some $\eps\in]0,1]$ such that $\delta+\eps<
\fr12$. Then,
$$
 \frac{1}{z-a}\in S_{\delta +\eps}^{\eps}\,.
$$
\end{lem}

\begin{proof}
The hypothesis  $|\Im z|\geq h^{\eps}$ implies immediately that
\begin{equation}\label{e:imz}
\fr{1}{|z-a|}\leq h^{-\eps}\,.
\end{equation}
To control the derivatives of $(z-a)^{-1}$, we will make use of the Fa\`a di Bruno
formula \cite{com}. For $n\geq2$,  Let $\alpha=(\alpha_1,...,\alpha_n)\in \N^n$ be a multi-index, and  $\Pi$ be the set of partitions of the ensemble
$\{1,...,|\alpha|\}$. For $\pi\in\Pi$, we write $\pi=\{B_1,...B_r\}$, where
$B_i$ is some subset of $\{1,...,|\alpha|\}$ and can then be seen as a
multi-index. Here  $|\alpha|\geq r\geq 1$, and we denote $|\pi|=r$. For two
smooth functions $g:\R^n\mapsto \R$ and $f:\R\mapsto\R$ such that $f\circ g$ is well defined, one has
\begin{equation}\label{e:fdb}
\d^\alpha f\circ
g=\sum_{\pi\in\Pi}\d^{|\pi|}f(g)\prod_{B\in\pi}\fr{\d^{|B|}g}{\prod_{j\in
  B}\d x_{\alpha_j}}\,.
\end{equation}
We now take $n=2$, and $\alpha\in \N^2$.  Using this formula for
$\d^{\alpha}\fr{1}{z-a}$ and recalling that $a\in S_\delta^{0}$, we
get for each partition $\pi\in\Pi$  a sum of terms which can be
written :
$$
\fr{1}{(z-a)^{r+1}}\prod_{B\in\pi}\d^{B}(z-a)\lesssim
h^{-\eps(r+1)}h^{-\delta |\alpha|}\,.
$$
Since we have $r\leq|\alpha|$, this concludes the proof.
\end{proof}

The preceding lemma allows us to obtain now a useful resolvent
estimate :
\begin{lem}\label{l:resol} Choose $\eps< \frac{1-2\delta}{4}$.
Suppose as above that $a\in S_\delta^0$ is real  and $z\in\Omega$ with
$|\Im z|\geq h^{\eps}$. Then,
$$
\cR_z(a)=\opn\left( \frac{1}{z-a}\right) + R_h(z)\,,
$$
where $R_h(z)\in \li(\hn)$ satisfies $
\norm{R_h(z)}=\cO(h^{1-2(\delta+2\eps)})\, $, uniformly in $z$.
\end{lem}

\begin{proof} We will denote $\Ohhnz(q)$ an operator which depends continuously on $z$ and whose norm in $\hn$ is of order $q$. 
By the preceding lemma and the symbolic calculus \eqref{e:symbcalc}, for $|\Im z| \geq h^{\eps} $ we can write :
$$
\opn(z-a)\opn\left(\fr{1}{z-a}\right)=\Id -\opn(r_z)
$$
with $\opn(r_z)=\Ohhnz(h^{1-2\delta-3\eps})$. For $h$ small enough, the
right hand side is invertible :

\begin{equation}\label{e:neum}
(\Id-\opn(r_z))^{-1}=\Id +\Ohhnz(h^{1-2\delta-3\eps})\,.
\end{equation}
We now remark that the G\aa rding inequality implies
$\opn\lp\fr{1}{z-a}\rp=\Ohhnz(h^{-\eps})\,$. Since we have obviously
$$
\cR_z(a)=\opn\left(\fr{1}{z-a}\right)(\Id-\opn(r_z))^{-1}\,,
$$
we obtain using Eq. \eqref{e:neum}
$$
\cR_z(a)=\opn\lp\fr{1}{z-a}\rp + \Ohhnz(h^{1-2\delta-4\eps})\,.
$$
It is now straightforward to check that all the remainder terms are
uniform with respect to $z$, since we have $z\in\Omega$ and $|\Im
z|\geq h^\eps$.
\end{proof}

We can now formulate a first order functional calculus for $a\in
S_\delta^0$.
\begin{prop}[Functional calculus]\label{p:calcfunct}
Let $a\in S_\delta^0$ real, and  $\fc$ a $w-$admissible
function with $w(h)^{-1}\lesssim h^{-\eta}$ such that
\begin{equation}\label{e:condition}
0\leq \eta<\fr{1-2\delta}{6}\,.
\end{equation}
Then, for any $\eps>0$ such
that
$
\eta<\eps<\fr{1-2\delta}{6}\,,
$
we have :
$$
\fc(\opn(a))=\opn(\fc(a))+ \Ohh(h^{1-2\delta-6\eps})
$$
and $\fc(a)\in S_{\eta+\delta}^0$.
\end{prop}
\begin{proof}
Let us write  the Lebesgue measure $dxdy=\frac{d\bar z
\wedge dz}{2i}$. From the operator theory point of view, we already
know that for any bounded self-adjoint operator $A$ 
$$
\fc(A)=\frac{1}{2\i\pi}\int_\C \dbar \tfc(z,\bar z)
(z-A)^{-1}\,d\bar z\wedge dz\,,
$$
see for example \cite{ds} for a proof.

Because of the condition expressed by Eq. \eqref{e:condition}, it is
possible to choose $\eps>0$ such that $\eta<\eps<\fr{1-2\delta}{6}$.
We now divide the complex plane into two subsets depending on $\eps$: $\C=\Omega_1\cup \Omega_2$, where $\Omega_1=\{z\in\C : |\Im z|<
h^{\eps}\}$ and $\Omega_2=\C\setminus\Omega_1$. We first treat the
case $z\in\Omega_1$. We recall that for $z\notin\R$, one has
\begin{equation}\label{e:rs-ineq}
\norm{\cR_z(a)}\leq\fr{1}{|\Im z|}\quad\trm{and}\quad \fr{1}{|z-a|}\leq \fr{1}{|\Im z|}\,.
\end{equation}
Using Eqs. \eqref{tf1} and \eqref{e:rs-ineq}, we obtain for any
$m\in\N$ :
\begin{eqnarray*}
 \lno \int_{\Omega_1} \dbar \tfc(z,\bar z) \cR_z(a)\,d\bar z\wedge dz
\rno
 &\leq&\int_{\Omega_1\cap\Supp(\tfc)}\frac{1}{|\Im z|}C_m|\Im z|^m w(h)^{-m-2}d\bar z\wedge dz\\
&\lesssim&h^{(m-1)\eps}w(h)^{-m-2}\,.
\end{eqnarray*}
Since $h^{\eps} w(h)^{-1}\lesssim h^{\eps-\eta}\hto0 0$, by taking $m$
sufficiently  large we obtain 
$$\lno\int_{\Omega_1} \dbar \tfc
\cR_z(a)\,d\bar z\wedge dz\rno=\cO(h^{\infty})\,.
$$
We now consider $z\in\Omega_2$. Using Lemma \ref{l:resol}, we obtain
\begin{eqnarray}
\int_{\Omega_2} \dbar \tfc(z,\bar z) \cR_z(a)\,\fr{d\bar z\wedge
dz}{2\i\pi}&=&\int_{\Omega_2}\dbar \tfc(z,\bar z)
\opn\left(\frac{1}{z-a}\rp \fr{d\bar z\wedge dz}{2\i\pi} \nonumber\\
& &+ \int_{\Omega_2}\dbar \tfc(z,\bar z) R_h(z)\fr{d\bar z\wedge
dz}{2\i\pi}\nonumber\\
&=&\trm{I}\, +\, \trm{II}\,.\label{e:funct0}
\end{eqnarray}
To rewrite the first term I in the right hand side, we use the
Fourier representation \eqref{e:fourier}:
\begin{eqnarray*}
\trm{I}&=& \fr{1}{2\i\pi}\int_{\Omega_2}\dbar \tfc \sum_{\mu,\nu\in
\Z}  \hat T_{\mu ,\nu } c_{\mu ,\nu }(z) d\bar z\wedge
dz\\
&=&\sum_{\mu,\nu\in \Z} \hat T_{\mu ,\nu }\int_{\t2} dxdy
e^{2i\pi(\mu x-\nu y)} \fr{1}{2\i\pi}\int_{\Omega_2}
\fr{\dbar\tfc(z,\bar z)}{z-a(x,y)}d\bar z\wedge dz\\
&=&\opn\lp \fr{1}{2\i\pi}\int_{\Omega_2} \fr{\dbar \tfc(z,\bar
z)}{z-a}d\bar z\wedge dz\rp\,,
\end{eqnarray*}
where we used the uniform  convergence of the Fourier series
\eqref{e:fourier}, the Fubini Theorem and the linearity of the
quantization $f\mapsto \opn(f)$. Note that
$$\int_{\Omega_2} \frac{\dbar \tfc(z,\bar z)}{z-a}d\bar z\wedge dz =\int_{\C}
\frac{\dbar \tfc(z,\bar z)}{z-a}d\bar z\wedge dz -\int_{\Omega_1}
\frac{\dbar \tfc(z,\bar z)}{z-a}d\bar z\wedge dz\,.$$
Now, Eqs. \eqref{e:rs-ineq} and \eqref{tf1} yields to
$$
\int_{\Omega_1} \frac{\dbar \tfc(z,\bar z)}{z-a}d\bar z\wedge dz=\cO(h^\infty)\,,
$$
and since the Cauchy formula for $C^{\infty}$ functions
implies that
$$
\fr{1}{2\i\pi}\int_{\C} \frac{\dbar \tfc(z,\bar z)}{z-a}d\bar
z\wedge dz =\tfc|_{\R}(a)=\fc(a)\,,
$$
we can simply write
\begin{equation}\label{e:funct1}
\trm{I}=\opn\left(\fr{1}{2\i\pi}\int_{\Omega_2} \frac{\dbar
\tfc(z,\bar z)}{z-a}d\bar z\wedge dz\right)
=\opn(\fc(a))+\Ohh(h^{\infty})\,.
\end{equation}
For the  term II in Eq. \eqref{e:funct0}, we first remark that
$\dbar  \tfc$ is compactly supported. Then, we use Eq. \eqref{tf1}
and Lemma \ref{l:resol} to write :
\begin{eqnarray}\label{e:funct2}
 \lno\int_{\Omega_2}\dbar \tfc(z,\bar z) R_h(z)d\bar z\wedge dz
\rno &\leq& \int_{\Omega_2} \norm{R_h(z)}\,|\dbar \tfc (z,\bar z) |\,
d\bar z\wedge dz  \nonumber\\
&\lesssim&  h^{1-2\delta-4\eps}h^{-2\eta}\leq h^{1-2\delta-6\eps}\,.
\end{eqnarray}
Hence, combining Eqs. \eqref{e:funct0},\eqref{e:funct1} and
\eqref{e:funct2},
\begin{eqnarray*}
\frac{1}{2\i\pi}\int_\C \dbar \tfc(z,\bar z) \cR_z(a)\,d\bar z\wedge
dz&=&\opn(\fc(a))+\Ohh( h^{1-2\delta-6\eps})\,,
\end{eqnarray*}
as was to be shown. The last assertion is a direct application of Eqs. \eqref{e:fdb} and  \eqref{h:countfunct}.
\end{proof}

Below, we will  frequently  make use of the functional calculus for some
perturbed operators. Our main tool for this purpose is stated as
follows:
\begin{cor}[Functional calculus -- perturbations]\label{c:calcfunctpert}
Let $a\in S_\delta^0$ be a real symbol and  $\fc$ be a $w-$admissible
function with $w(h)^{-1}\lesssim h^{-\eta}$ and $0\leq \eta<\fr{1-2\delta}{6}$.
Consider $B_h\in\li(\hn)$, with the properties :
$$
\begin{array}{ll}
(i) & \opn(a)+B_h \trm{ is self-adjoint}\,,\\
(ii) & \|B_h\|=\cO(h^{\nu}) \trm{for some } \nu>4\eta>0\,. \\
\end{array}
$$
Then, at first order, the functional calculus is still valid : for
any $\eps>0$ such that
$$
\eta<\eps<\fr{1-2\delta}{6}\quad\trm{and}\quad \eta<\eps<\fr{\nu}{4}\,,
$$
we have :
$$
\fc(\opn(a)+B_h)=\opn(\fc(a))+\Ohh(h^{\min(1-2\delta-6\eps,
\nu-4\eps)})
$$
and $f_w(a)\in S_{\eta+\delta}^0$.
\end{cor}
\begin{proof}
We must find the resolvent of $\opn(a)+B_h$. From $(ii)$ we can choose $\eps>0$ such
that $\eta<\eps<\fr{1-2\delta}{6}$ and  $\eta<\eps<\fr{\nu}{4}$. Note that  this implies
$\nu-4\eps>0$. As in Lemma \ref{l:resol}, we choose a
compact domain $\Omega\subset\C$ with $z\in\Omega$, and split $\Omega=\Omega_1\cup\Omega_2$ as above. Suppose first that $z\in \Omega_2$, i.e. $|\Im(z)|\geq
h^\eps$. Then, using the condition $(ii)$ we get
$$
\Id - (\opn(z-a)-B_h)\opn(\fr{1}{z-a})=  \opn(r_z) +
B_h\opn(\fr{1}{z-a})\,.
$$
The norm of the right hand side is of order
$h^{\min(1-2\delta-3\eps,\nu-\eps)}$ uniformly for $z\in\Omega_2$, hence for $h$ small enough, we can use the same method  employed
in the Lemma \ref{l:resol} to get :
\begin{eqnarray*}
(z-\opn(a)-B_h)^{-1}&=& \opn(\fr{1}{z-a})\lp\Id-\opn(r_z) -
B_h\opn(\fr{1}{z-a})\rp ^{-1}\\
&=& \opn( \fr{1}{z-a} ) + R_h
\end{eqnarray*}
where $\norm{R_h}\lesssim h^{\min(1-2\delta-4\eps,\nu-2\eps)}$
uniformly in $z$.  The next steps are now exactly the same as above,
and we end up with
$$
\fc(\opn(a)+B_h)=\opn(\fc(a))+  \int_{\Omega_2}\dbar \tfc(z,\bar z)
R_h(z)\fr{d\bar z\wedge dz}{2\i\pi}
 +\Ohh(h^{\infty})\,,
$$
where now $\| R_h(z)\|\lesssim h^{\min(1-2\delta-4\eps,\nu-2\eps)}$. Since $\|\dbar \tfc\|_{C^0}\lesssim h^{-2\eta}$, 
this concludes the proof.
\end{proof}

As a direct application of the last two proposition, let us show an
analogue of the spectral Weyl law on the torus.
\begin{prop}\label{p:Wlaw}
Let $a\in S_0^0$, and  $B_h\in\cL(\hn)$ as in
\corref{c:calcfunctpert}.   Choose $\eps>0$ arbitrary small but fixed, and call $A_h=\opn(a)+B_h$  Let $E_1,E_2$ be positive numbers,  $I\defeq[E_1,E_2]$ and $I_\pm\defeq[E_1\mp \eps,E_2\pm \eps]$. Then,
\begin{equation}
\int_{\t2} \bbbone_{I_-}(a) d\mu +o_h(1) \leq 
 h\#\{\lambda\in\Spec A_h \cap I\} \leq \int_{\t2} \bbbone_{I_+}(a)+o_h(1)\,.
\end{equation}
\end{prop}
\begin{proof} 
Define a smooth function $\chi^+$  such that for some $C>0$,
$\chi^+(x)=1$ if $x\in I$ and $\chi^+(x)=0$ if $x\notin I_+$. Define
as well $\chi^-=0$ outside $I$  and $\chi^-=1$ on $I_-$. Denote
$D_h=h\#\{\lambda\in\Spec A_h \cap I\}$. Then,
$$
h\Tr(\chi^-(A_h))\leq D_h \leq h\Tr (\chi^+(A_h))\,.
$$
By the \corref{c:calcfunctpert}, $h\Tr (\chi^+(\opn(a)+B_h))=\int_{\t2} \chi^+(a) d\mu + \cO(h^\alpha)$ for some $\alpha>0$, and  $h\Tr(\chi^-(\opn(a)+B_h))=\int_{\t2} \chi^-(a)d\mu + \cO(h^\alpha)$ as well. But obviously, 
$$
\int_{\t2} \bbbone_{I_-}(a) d\mu \leq  \int_{\t2} \chi^-(a)\quad\trm{and}\quad  \int_{\t2} \chi^+(a) \leq \int_{\t2} \bbbone_{I_+}(a)\,.
$$
This yields to
$$
\int_{\t2} \bbbone_{I_-}(a) d\mu +o_h(1)\leq D_h\leq \int_{\t2} \bbbone_{I_+}(a)+o_h(1)\,. 
$$
\end{proof}

\section{Eigenvalues density}\label{s:density}

\subsection{The operator $\Sn$}\label{s:Sn}
Consider our damping function $a\in S_0^0$, with $\Ran
(|a|)=[a_-,a_+]$, $a_- > 0$, $a_+\leq 1$. To simplify the following
analysis, we will suppose without any loss of generality that
$a_+=1$. As mentioned above, to study the radial distribution of
$\Mak$ it will be useful to first consider the sequence of operators

\begin{equation}\label{e:tsn}
\tsn(a)\defeq \Mak^{\dg n}\;\Mak^n,\qquad n\geq 1\,.
\end{equation}

Let us show that for $h\to 0$ and $n\geq 1$ possibly depending on $h$,  these operators
can be rewritten into a more simple form, involving $n-$time evolutions of the observable $a$ by the map $\ka$. Using the composition of operators
\eqref{e:symbcalc} and the Egorov property \eqref{e:Egorov}, we will
show that the quantum to classical correspondence is valid up to
times of order $\log h^{-1}$, as it is usually expected. In what
follows, for any constant $C$, we will make use of the notation
$C^{\pm}=C\pm\eps$ with $\eps>0$ arbitrary small but \emph{fixed} as
$h\to 0$. We allow the value of $\eps$ to change from equation to
equation, hence $C^\pm$ denotes any constant arbitrary close to $C$,
$C^+$ being larger than $C$ and $C^-$ smaller: 
$\eps$ will then  be chosen small enough so that the equations where $C^\pm$
appear are satisfied. We also recall the following definitions, already introduced in Theorem \ref{th:large}:
\begin{equation}\label{e:ehren}
\Gamma=\log(\sup_x \norm{D\ka|_x})\quad \trm{and} \quad  T_{a,\ka}=\fr{1}{2\Gamma-12\log a_-}\,.
\end{equation}
\begin{prop}\label{p:sn} Let $\tau>0$ be a constant such that $\tau <T_{a,\ka}$. If $E(x)$ denotes the integer part of $x$,  define $n_\tau=E(\tau\log h^{-1})$. Then, for $n\leq n_\tau$, the
operators
\begin{align*}
&\mathcal S_n(a)= \lp\Mak^{\dg n}\;\Mak^n \rp^{\fr{1}{2n}}\\
&\ell \mathcal S_n(a)=\log \mathcal S_n(a)
\end{align*}
are well defined if and only if $\operatorname{Ker} (\Mak) =0$. Furthermore, if we set
$$
a_n\defeq \prod_{i=1}^n |a\circ\ka^i|^{\fr{1}{n}}\,,\quad \lan\defeq
\fr{1}{n}\sum_{i=1}^n \log |a\circ\ka^i|\,,
$$
we have for $n\leq n_\tau$:
\begin{align}
&\lsn(a)=\opn(\lan) +\Ohh(h^{\sigma^-})\label{e:lsn}\\
&\Sn(a)=\opn(a_n) + \Ohh(h^{\sigma^-})\label{e:sn}
\end{align}
where $\sigma=1-\tau/T_{a,\ka}>0$.
\end{prop}
\begin{proof} Let us underline the main steps we will encounter below. Writing first 
$$ \ti a_n \defeq \prod_{i=1}^n |a\circ\ka^i|^2 $$for $n\in\N$, 
we show that for $\tau<\fr{1}{2\Gamma}$, $\ti a_n$ belongs to a symbol
class $S_{\delta^+}^0$ with $\delta^+<1/2$. Then, we show by using the
symbolic calculus \eqref{e:symbcalc} and the Egorov property \eqref{e:Egorov} that $\ti
\Sn(a)=\opn(\ti a_n) + \cO(h^{\nu^-})$ for some $\nu^->0$. Finally, by bounding the spectrum of
$\ti \Sn(a)$, we will complete the proof of the proposition by using
the functional calculus  to define and compute both  $\lsn(a)$ and $\Sn(a)$. 
\begin{lem}\label{l:class1} 
Let $a\in S_0^0$ be a symbol on the torus. Then, for any multi-index $\alpha\in\N^2$, there exists $C_{\alpha,a,\ka}>0$ such that 
\begin{equation}\label{e:class1}
\forall n\geq 1\,,\quad \|\d^\alpha (a\circ\ka^n)\|_{C^0}\leq C_{\alpha,a,\ka}\e^{n|\alpha|\Gamma}\,.
\end{equation}
Hence for any  $\tau>0$ such that $\tau<\fr{1}{2\Gamma}$, we have uniformly for $n\leq n_\tau$:
\begin{equation}\label{eq:tian}
a\circ\ka^n\in S_\delta^0\,,\quad \delta=\tau\Gamma<\fr12.
\end{equation} 
\end{lem}
\begin{proof}
The behavior expressed by Eq. \eqref{e:class1} is well known for flows \cite{br}, and we refer to \cite{fnw}, Lemma 1 for a detailed proof in the case of applications. The second part of the lemma follows easily from \eqref{e:class1}.
\end{proof}

\begin{lem}\label{l:tan}
If $\tau<\frac{1}{2\Gamma}$, we have : 
$$
\forall n \leq n_\tau\,,\quad\ti a_n \in S_{\delta^+}^0\quad \trm{with}\  \delta^+=\tau\Gamma^+<\fr12\,.
$$
\end{lem}
\begin{proof} Since $a_+=1$ and $0<a_-<1$, $\ti a_n$ is uniformly bounded from above with respect to $h$. It is then  enough to show that for every multi index
$\alpha\in\N^{2}$ and $n\leq n_\tau$, one has 
$$\d^\alpha \prod_{i=1}^n (a\circ\ka^i)\lesssim 
h^{-\delta^+|\alpha|}.
$$
Set by convention $\d^0 f = f$. Applying the Leibniz rule,  we can
write
$$
\lno \d^\alpha \prod_{i=1}^n
a\circ\ka^i\rno_{C^0}\leq n^{|\alpha|}\sup_{\alpha_1+...+\alpha_n=\alpha} \prod_{i=1}^n
\|\d^{\alpha_i}(a\circ\ka^i)\|_{C^0}\,.
$$
Let us look at a typical term in the product appearing in the right hand side. Since at most $|\alpha|$
indices $\alpha_i$ in are non zero and $|a\circ\ka^i|\leq1$, 
\begin{align*}
\prod_{i=1}^n \|\d^{\alpha_i}(a\circ\ka^i)\|_{C^0}&\leq
\prod_{i=1}^{n}
C_{\alpha_{i},a,\ka} \e^{i |\alpha_{i}| \Gamma }\\
&\leq (\sup_{|\beta|\leq|\alpha|}
C_{\beta,a,\ka})^{|\alpha|}\e^{n|\alpha|\Gamma}\defeq
K_{\alpha,a,\ka}\e^{n|\alpha|\Gamma}\,.
\end{align*}
Finally, we simply get
$$
\lno \d^\alpha \prod_{i=1}^n a\circ\ka^i\rno_{C^0}\leq
K_{\alpha,a,\ka}\e^{|\alpha|(\Gamma n + \log n)}\leq
K_{\alpha,a,\ka}\e^{|\alpha|\Gamma ' n}
$$
for some $\Gamma'>\Gamma$. If we choose $\Gamma'=\Gamma^+$ and
$\Gamma^+-\Gamma$ small enough  such that $\tau\Gamma^+<\fr12$, the last equation will be true only for $n\geq n_0$, with $n_0$ fixed independent of $h$. 
Since  for $n<n_0$, we obviously have $\ti a_n \in S_0^0$,  we finally conclude that if  $\tau\Gamma^+<\fr12$, $\ti a_n\in   S_{\delta^+}^0$ with $\delta^+=\tau\Gamma^+<\fr12$, uniformly for $n\leq n_\tau$.
\end{proof}
We can now rewrite more explicitly Eq. \eqref{e:tsn}.
\begin{lem}\label{l:iter}
Choose $\tau>0$ small enough such that
$\delta=\tau\Gamma<\fr{1}{2}$, and take as before $n\leq n_\tau$.
Then, we have
\begin{equation}\label{e:iter}
\tsn(a)=\opn(\ti a_n) + \Ohh(h^{\nu^-})
\end{equation}
where $\nu=1-2\delta$.
\end{lem}
\begin{proof}
The Lemmas \ref{l:class1} and \ref{l:tan} tell us that uniformly for $n\leq
E(\tau\log h^{-1})$, the symbols $a\circ\ka^n$ and $\ti a_n$ belong to the class $S_{\delta^+}^0$ if $\delta=\tau\Gamma<\fr12$. This allows to write (using $U\equiv U_h(\ka)$ for simplicity):

\begin{align}
U^\dg\opn(\bar a)\opn(a)U&=U^\dg\opn(\bar a\sharp_h a)U\nonumber\\
&=U^\dg\opn(|a|^2)U+\Ohn(h^{1-2\delta^+})\label{e:ego2}\\
&=\opn(|a\circ\ka|^2) + \Ohn(h^{1-2\delta^+})\label{e:ego3}\\
&=\opn(\ti a_1) + R_1^1\,,\nonumber
\end{align}
where  the symbolic calculus \eqref{e:symbcalc} has been used to get
\eqref{e:ego2} and Eq. \eqref{e:Egorov} to deduce
Eq. \eqref{e:ego3}. The remainder $R_1^1$ have a norm of 
order $h^{1-2\delta+}$ in $\hn$. This calculation  can be iterated : suppose that the preceding step gave 
$$
\ti \cS_k(a)=\opn(\ti a_k) + \sum_{i=1}^k R^k_i\,,
$$
with $R_i^k=\Ohn(h^{1-2\delta^+})$ for $1\leq i\leq k$. Then, with the same arguments as those we 
used for the first step, we can find $R_{k+1}^{k+1}=\Ohn(h^{1-2\delta^+})$ such that 
$$
U^\dg \opn(\bar a)\opn(\ti a_k)\opn(a)U=\opn(\ti a_{k+1})+R_{k+1}^{k+1}\,.
$$
If we define now $R^{k+1}_i = U^\dg\opn(\bar a)R^{k}_i\opn(a)U=\Ohn(h^{1-2\delta^+})$, we get 
\begin{eqnarray*}
U^\dg\opn(\bar a)\ti\cS_k (a) \opn(a)U&=&U^\dg\opn(\bar a)\opn(\ti a_k)\opn(a)U
+ \sum_{i=1}^k R^{k+1}_i \\
&=&\opn(\ti a_{k+1})+ \sum_{i=1}^{k+1} R^{k+1}_i\,.
\end{eqnarray*}
This shows that
\begin{align*}
\ti\Sn(a)&=\opn(\ti a_n) + n\,\Ohh( h^{1-2\delta^+})\\
&=\opn(\ti a_n) + \Ohh( h^{1-2\delta^+})\,.
\end{align*}
In the preceding equation, the second line comes from the fact that $n\lesssim \log h^{-1}$. 
 Since we defined $\nu=1-2\delta$,  the
lemma is proved.
\end{proof}

From now on, we will always assume $n\leq n_\tau$ for some
$\tau>0$ \emph{fixed}. We will also choose
$\tau$ small enough such that
\begin{equation}\label{e:tauchoice}
\tau < T_{a,\ka}=\fr{1}{2\Gamma-12\log a_-}\,.
\end{equation}
This condition ensures in particular that  Lemma \ref{l:iter} is
valid, but it turns out that we will need a stronger condition than $\tau\Gamma<1/2$ to 
complete the proof of Proposition \ref{p:sn}. 

In order to apply the
functional calculus to $\tsn (a)$, we must bound its spectrum with
the help of Proposition \ref{p:gard}. This is expressed in the following
\begin{prop}\label{p:boundspec}
Define $\eta=-2\tau\log a_->0$. For $h>0$ small enough, we have
$$
\Spec \tsn(a) \subset [C_{n,h}, 2]\,,
$$
where $C_{n,h}=\e^{2n\log a_-}-ch^{1-2\delta^+}$ and $c>0$. In particular, there exists a constant $C>0$ such that $C_{n,h}\geq Ch^\eta$, and $\tsn(a)$ has strictly positive spectrum.
\end{prop}
\begin{proof}
We begin by proving the result for $\opn(\ti a_n)$. Since we have for some $c_1>0$
\begin{align*}
\ti a_n\geq a_-^{2n}&=\e^{2n\log a_-}\geq c_1h^{-2\tau\log a_-}\,,
\end{align*}
we will consider the symbol $b_n=\ti a_n-\e^{2n\log a_-}\geq 0$.
From Lemma \ref{l:tan}, we have $b_n\in S_{\delta^+}^0$ and we can
apply Proposition \ref{p:gard} : for any
$\lambda\in\Spec(\opn(\ti a_n))$, there exists $c>0$ such that 
$
|\lambda|\geq \e^{2n\log a_-} -c h^{1-2\delta^+}
\geq  c_1h^{-2\tau\log a_-} - c h^{1-2\delta^+}\,.
$
In order to have a strictly positive spectrum, we must have
$-2\tau\log a_- < 1-2\delta^+$, which is satisfied if $\tau$ is
chosen according  to Eq. \eqref{e:tauchoice}. Hence, there is
a constant $C>0$ such that for $h$ small enough, $c_1h^{-2\tau\log a_-} - ch^{1-2\delta^+}> C h^{-2\tau\log a_-}$, and the lower bound is obtained. For the
upper bound, we remark that we assumed that
$a_+=1$. By Proposition \ref{p:gard}, it follows that any constant
strictly bigger than 1 gives an upper bound for the spectrum. We return now to $\tsn(a)$. Since we have $\|\tsn(a)-\opn(\ti a_n)\|=\cO(h^{\nu^-})$ and $\nu^->\eta$ thanks to Eq. \eqref{e:tauchoice}, we get the final result if $h$ is small enough.
\end{proof}
Let us finish now the proof of Proposition \ref{p:sn}. For $n\leq n_\tau$, we begin by  constructing a smooth
function $\chi_{n,h}$ compactly supported, equal to 1 on $\Spec(\ti\Sn)$. To do
this, we define
$$
\chi_{n,h}(x)=\left\{
\begin{array}{ll}
0& \textrm{ if } x\leq \fr12 C_{n,h}\\
1& \textrm{ if } x\in [C_{n,h},2]\\
0& \textrm{ if } x\geq 3
\end{array}
\right.
$$
Then, the function
$$
\ell_n : x\mapsto \fr{\chi_{n,h}(x)}{2n}\log x
$$
is smooth  and equal to the function $\fr{1}{2n}\log(x)$ on $\Spec
\tsn(a)$,  since we have shown that $ \Spec(\tsn(a) )\subset
[C_{n,h},2]$. Furthermore, $\ell_n$ is compactly supported,
uniformly bounded with respect to $h$. Since $C_{n,h}\geq Ch^\eta$, the function  $\chi_{n,h}$  can easily be chosen $h^{\eta}$--admissible, which means that $\ell_n$  will also be 
$h^{\eta}$--admissible. Applying the standard functional calculus,
we have
\begin{align*}
&\ell_n(\tsn(a))=\fr{1}{2n}\log \tsn(a)\defeq \lsn(a)\,,
\end{align*}
where the first equality follow from the fact that $\ell_n(x)=\fr{1}{2n}\log x$ on $\Spec(\tsn(a))$. To
complete the proof of Proposition \ref{p:sn}, we compute $\ell_n(\tsn(a))$
using both Lemma \ref{l:iter} and Corollary \ref{c:calcfunctpert}. For
this purpose, we must check that the conditions required by this
corollary are fulfilled. First, $\ti a_n \in S_{\delta^+}^0$, and from
Eq. \eqref{e:tauchoice} we have $\eta<\fr{1-2\delta^+}{6}$.
Second, 
the remainder in
Eq. \eqref{e:iter} is of order $h^{\nu^-}$ with $\nu=1-2\delta$, and we clearly have 
$$\fr{\nu^-}{4}>\fr{1-2\delta^-}{6}>\eta\,.$$
We now apply the Corollary  \ref{c:calcfunctpert} to obtain
\begin{equation}
\lsn(a)=\ell_n(\opn(\ti
a_n)+\Ohh(h^{\nu^-}))=\opn(\lan)+\Ohh(h^{r})\,,\ \ r>0.
\end{equation}
Let us show that $r=\sigma^-$, where
$\sigma=1-\tau/T_{a,\ka}$. Indeed, the functional calculus states
that 
$$
r=\min\{1-2\delta^+-6\eps,\nu^--4\eps\}
$$
where $\eps$ has to be chosen in the interval $]\eta,  \min(\fr{1-2\delta^+}{6},\fr{\nu^-}{4})[$. Let
us choose $\eps=\eta^+$.
 Since $\nu=1-2\delta$, 
we have
\begin{equation}\label{sigma}
r=\min\{1-2\delta^+-6\eta^+,\nu^--4\eta^+\}=1-2\delta^+-6\eta^+=(1-\tau/T_{a,\ka})^-=\sigma^-\,.
\end{equation}
To get now $\Sn(a)$, we again apply the standard functional
calculus, and we obtain 
$$
\exp\lp \ell_n(\tsn(a))\rp = \exp \lp\fr{1}{2n}\log \tsn \rp
\defeq \Sn(a)
$$
where the first equality follows again from the property of $\ell_n$ on $\Spec(\lsn)$. We already stress that this
equation is essential to study the eigenvalues distribution of
$\Mak$. Since $\ti a_n\in S_{\delta^+}^0$, $\lan=\fr{1}{2n}\log \ti a_n \in S_{\delta^+}^0$. 
To compute $\Sn(a)$,  we simply choose a smooth cutoff function
supported in a $h-$independent neighborhood of $\Spec(\lsn(a))$. The
Corollary \ref{c:calcfunctpert} is used again to get
\begin{equation}
\Sn(a)\defeq\exp(\lsn(a))=\opn(a_n)+\Ohh(h^{\sigma^-})\,.
\end{equation}
\end{proof}

\subsection{Radial spectral density}\label{s:radial} 
We begin by recalling some elementary properties of ergodic means with respect
to the map $\ka$.  Since $|a|>0$ on $\t2$, the function $x\mapsto \log
|a\circ\ka^i(x)|$ is continuous on $\t2$ for any $i\in\N$. The
Birkhoff ergodic theorem then states that
$\lim_{n\rightarrow\infty}\lan(x)\defeq \ell a_\infty(x)$ exists for
$\mu-$almost every $x$. More precisely, if we denote
$\eia\defeq\essi
a_\infty 
$ and
$
\esa\defeq\esss a_\infty$ for $a_\infty=\exp(\ell a_\infty)$, we have:
$$
\eia\leq a_\infty(x)\leq\esa \ \trm{ for }\mu-\trm{almost every }x\,.
$$
In particular, if $\ka$ is ergodic with respect to the Lebesgue
measure $\mu$, the Birkhoff ergodic theorem states that:
\begin{equation}\label{birk}
\textrm{For }\mu-a.e.\  x,\quad \lim_{n\rightarrow\infty}\log
a_n(x)=\int_{\TT^2} \log |a|\,d\mu
 = \log \moy a\,,
\end{equation}
and in this case, $\eia=\esa= \moy a$.

\emph{Proof of  Theorem.~\ref{th:strip}}. As in
Proposition~\ref{p:weylineq}, we order the eigenvalues of $\Mak$ and
$\Sn(a)$ by decreasing moduli. Take arbitrary small $\epsilon>0$,
$\delta>0$ and $0<\gamma\leq \fr{\epsilon}{3}\delta $. We recall
that $I_\delta= [\eia-\delta,\esa +\delta]$ and define
$\Omega_{n,\gamma^-}\defeq \TT^2\setminus a_n^{-1}(I_{\gamma^-})$, so
$$
\mu(a_n^{-1}(I_{\gamma^-}))=1-\mu( \Omega_{n,\gamma^-})\,.
$$
Eq.~\eqref{birk} implies that  $\lim_{n\rightarrow\infty}\mu(
\Omega_{n,\gamma^-})=0$. Hence, there exists $n_0\in \N$ such that
$$
n\geq n_0 \Longrightarrow \mu(\Omega_{n,\gamma^-})\leq
\fr{\epsilon\,\delta}{6a_+- 3\moy a}= \fr{\epsilon\,\delta}{6- 3\moy a}\,.
$$
We now choose some  $n\geq n_0$.  Applying Proposition ~\ref{p:Wlaw} to
the operator $\Sn(a)$, we get immediately: 
\begin{equation}\label{e:weyltorus}
h\,\#\Big\{s\in \Spec (\Sn(a)) : s\in I_{\gamma}\Big\}\geq  
\mu(a_n^{-1}(I_{\gamma^-}) )+o_h(1)\,. 
\end{equation} 
Using the Weyl inequalities,
we can relate the spectrum of $\Sn(a)$ with that of $\Mak$. Call
$d_h\defeq \#\{\lambda\in \Spec(\Mak)\: :\: | \lambda|> \esa +
\delta\}$, where eigenvalues are counted with their algebraic
multiplicities. Hence, using Corollary~\ref{c:Weyl-ineq}, we get
\begin{equation}\label{samedisc} 
d_h\,\big(\esa + \delta\big)\leq \sum_{k=1}^{d_h}
|\lambda_k|\leq\sum_{k=1}^{d_h} s_k\,. 
\end{equation}
Among the
$d_h$ first (therefore, largest) eigenvalues
$(s_i)_{i=1,\ldots,d_h}$ of $\Sn(a)$, we now distinguish those
which are larger than $\gamma +\esa$, and call $d_h'=\#\{1\leq
i\leq d_h,\; s_i\leq\gamma +\esa\}$ the number of remaining ones.

Applying Proposition \ref{p:gard} to the observable  $a_n$, we are sure that 
for $h$ small  enough, $s_j\in\Spec(\Sn(a))\Rightarrow s_j < 2$. 
Hence, for $h$ small enough, \eqref{samedisc} induces
\begin{equation}\label{ineq1}
 d_h\delta\leq d_h'\gamma + (d_h-d_h')(2 -
\esa)\,. 
\end{equation}
 By construction,
$$
d_h-d_h'
\leq  \#\Big\{s \in \Spec (\Sn(a)) : s\notin I_{\gamma} \Big\}\,,
$$
and using \eqref{e:weyltorus}, we deduce 
$$
d_h-d_h'\leq h^{-1}(\mu( \Omega_{n,\gamma^-})+o_h(1))\,.
$$
Dividing now  \eqref{ineq1} by $h^{-1}\delta$ and using the preceding equation, we obtain :
\begin{eqnarray*}
h\,d_h&\leq & \fr{\gamma}{\delta} + \fr{2-\moy a}{\delta}\mu( \Omega_{n,\gamma^-}) + o_h(1)\\
&\leq&  \fr{\epsilon}{3} + \fr{\epsilon}{3} +o_h(1)\,.
\end{eqnarray*}
To get the second inequality we used the condition on $\gamma>0$.
There exists $h_0(\epsilon,n)$ such that for any $h\leq h_0(\epsilon,n)$,
the last remainder is smaller than $\epsilon/3$.
We conclude that  $\forall \epsilon,\delta >0$,  $\exists h_0^{-1}\in \N$
such that: 
\begin{equation}\label{ineq2} 
h\leq h_0 \Rightarrow \#\Big\{\lambda
\in \Spec (\Mak) \;:\; |\lambda|> \esa + \delta \Big\}\leq
h^{-1}\epsilon\,. 
\end{equation}
To estimate the number of eigenvalues
$|\lambda|<\eia-\delta$, we use the inverse propagator:
\begin{equation}\label{e:inverse}
\Mak^{-1}=U_h(\kappa)^{-1}\,\opn(a)^{-1}=
\opn(a^{-1}\circ\kappa)\,U_h(\kappa)^{-1} + \Ohh(h)\,,
\end{equation}
where Eq. \eqref{e:Egorov} has been used for the last equality. 
Since $U_h(\kappa)^{-1}$ is a quantization of the map $\kappa^{-1}$,
the right hand side has a form similar with \eqref{Qmap}, with a small perturbation of order $h$. The
function $a^{-1}$ satisfies
$$
a_+^{-1}\leq |a^{-1}|\leq a_-^{-1}\quad\text{ and }\quad
a^{-1}_\infty=(a_\infty)^{-1}\,.
$$
Hence, we also have $\operatorname{EI}_\infty(a^{-1})=
(\operatorname{ES}_\infty(a))^{-1},\ \
\operatorname{ES}_\infty(a^{-1})=
(\operatorname{EI}_\infty(a))^{-1}\,.$ Applying the above results to
the operator $\Mak^{-1}$, we find some $h_1(\epsilon,\delta)>0$ such that
\begin{equation}\label{ineq3} 
h\leq h_1 \Rightarrow \#\Big\{\lambda \in \Spec
(\Mak)\; :\; |\lambda|< \eia - \delta \Big\}\leq h^{-1}\epsilon\,. 
\end{equation}
Grouping  (\ref{ineq2}) and  (\ref{ineq3}) and taking $\epsilon$
arbitrarily small, we obtain Eq. \eqref{e:strip}\,. \stopthm

\subsection{Angular density}\label{s:angular}
From Eq. \eqref{e:bounds}, we know that for $h$ small
enough, all eigenvalues of $\Mak$ satisfy $|\lambda_i|\geq
a_-/2$. We can then write these eigenvalues as
$$
\lambda_i =r_i\, e^{2i\pi\theta_i}\,,\quad
r_i=|\lambda_i|,\quad \theta_i\in \SP^1\equiv [0,1)\,.
$$
(we skip the superscript $(h)$ for convenience). We want to show
that, under the ergodicity assumption on the map $\kappa$, the arguments
$\theta_i$ become homogeneously distributed over $\SP^1$ in the
semiclassical limit. We adapt the method presented in
\cite{bdb2,mok} to show that the same homogeneity holds for the
eigenangles of the map $U_h(\ka)$. The main tool is the study of the
traces $\Tr(\Mak^n)$, where $n\in\N$ is taken arbitrary large but
independent of the quantum dimension $h^{-1}$.

\emph{Proof of \thmref{th:ergo}, (ii)}. We begin by showing a useful
result concerning the trace of $\Mak^n$:
\begin{prop}\label{p:traces}
The ergodicity assumption on $\kappa$ implies that
\begin{equation}\label{e:traces}
\forall n\in \Z\setminus \{0\},\qquad
\lim_{h\rightarrow0} h\Tr \Mak^n=0\,.
\end{equation}
\end{prop}

\begin{proof} Note here that $n$ is fixed, independently of $h$. As
  before, we write $U\equiv U_h(\ka)$ for convenience. Inserting
  products $U^{-1}U$ and using Egorov's property,  we get:
\begin{equation}\label{e:Makn}
\Mak^n=\opn(a'_n)U^n+\cO_n(h)\,,\quad\text{where}\quad
a'_n\defeq \prod_{j=0}^{n-1}a\circ\ka^{-j}\,. 
\end{equation}
Let $\epsilon_0>0$.
The next step consists of exhibiting a finite  open cover
$\t2=D_0\cup\bigcup_{i=1}^M D_i$ with the following properties :
\begin{itemize}
\item $D_0$ contains the fixed points of $\ka^n$ and has Lebesgue measure
$\mu(D_0)\leq \epsilon_0$. Such a set can be constructed thanks to
Proposition~\ref{p:Minkowski}.
\item For each $i=1,\ldots,M$, $\ka^n(D_i)\cap D_i=\emptyset$. This is possible,
because $\ka^n$ is continuous and without fixed points on
$\t2\setminus D_0$.
\end{itemize}
Then, a partition of unity subordinated to this cover can be
constructed:
$$
1=\chi_0+\sum_{i=1}^{M}\chi_i\,,
$$
with $\chi_i\in C^{\infty}(\t2)$, $\Supp \chi_i\subset D_i$
for $0\leq i\leq M$. Notice that the condition on $D_i$, $i\geq 1$
implies that $\chi_i (\chi_i\circ\kappa^{-n})\equiv 0$. After quantizing this partition we write:
$$
\Mak^n=\opn(\chi_0)\Mak^n+\sum_{i=1}^{M}\opn(\chi_i)\,\Mak^n\,.
$$
Using \eqref{e:Makn} and taking the trace:
\begin{eqnarray*}
h \Tr (\Mak^n)&=&h\Tr\big(\opn(\chi_0)\opn(
a'_n)U^n\big)\\
& & + h\sum_{i=1}^{M} \Tr\big(\opn(\chi_i)\opn(a'_n)U^n\big) +\cO_n(h)\,.\\
\end{eqnarray*}
Since $\chi_0\geq 0$, $\opn(\chi_0)^\dg=\opn(\chi_0)$ and we can perform the first trace in the right hand side in the 
basis where $\opn(\chi_0)$ is diagonal. This yields to  
$$
|\Tr(\opn(\chi_0)\opn(a'_n)U^n)|\leq a_+^n| \Tr(\opn(\chi_0))|+\cO(1)\,.
$$
The  estimates \eqref{e:trace} and Eq.  \eqref{e:symbcalc} imply
that
$$
h\Tr (\opn(\chi_0)) =\int_{\t2} \chi_0\,d\mu + \cO(h)\,.
$$
Since $\mu(D_0)\leq\epsilon_0$,  the integral on
the right hand side satisfies $|\int_{\t2} \chi_0\,d\mu |\leq \epsilon_0$, and finally 
$h\Tr\big(\opn(\chi_0)\opn(
a'_n)U^n\big) = \cO(\epsilon_0)+\cO(h)$.

To treat the terms $i\geq 1$, we take first $\eps>0$ arbitrary small
and a smooth function
$\ti\chi_i\in C_0^\infty$ such that $\Supp \ti\chi_i\subset \Supp
\chi_i$ and  $\|\ti\chi_i^2-\chi_i\|_{C^0}=\cO(\eps)$. Then, we use the symbolic calculus to write
\begin{align*}
U^n \opn(\chi_i)&=U^n\opn(\ti\chi_i)^2 + \Ohh(h)+\Ohh(\eps)\\
&=U^n\opn(\ti\chi_i)U^{-n}\,U^n \opn(\ti\chi_i) +
\Ohh(h)+\Ohh(\eps)\,.
\end{align*}
Using the cyclicity of the trace, this gives
\begin{eqnarray*}
h\Tr(\opn(\chi_i)\opn( a'_n)\,U^n)
&=&h\Tr(U^n\opn(\ti{\chi_i})U^{- n}U^{n}\opn(\ti{\chi_i})\opn(
a'_n))\\
& &+ \cO(h)+\cO(\eps)\\
&=&h\Tr(\opn(\ti{\chi_i}\circ \kappa^{-n}) U^{n}\opn(\ti{\chi_i})\opn(
a'_n))\\
& &+ \cO(h)+\cO(\eps)\\
&=&h\Tr(\opn( a'_n\,\ti{\chi_i}\,(\ti{\chi_i}\circ \kappa^{-n}))
U^{n})+ \cO(h)+\cO(\eps)\\
&=&\cO(h)+\cO(\eps)\,.
\end{eqnarray*}
In the last line we used $\ti \chi_i (\ti \chi_i\circ\kappa^{-n})\equiv 0$.
Adding up all these expressions, we finally obtain 
\begin{equation} 
h\Tr (\Mak^n) = \cO(\epsilon_0)+\cO(\eps)+\cO(h)\,. 
\end{equation}
Since this estimate holds for arbitrary small $\epsilon_0$ and $\eps$, it proves
the proposition.
\end{proof}
We will now use these trace estimates, together with the information we already have on the radial spectral distribution
(Theorem.~\ref{th:ergo}, $(i)$) to prove the homogeneous
distribution of the angles $(\theta_i)$.

\medskip

{\em Remark:} This step is not obvious a priori: for a general
non-normal $h^{-1}\times h^{-1}=N\times N$ matrix $M$, the first few traces $\Tr(M^n)$
cannot, when taken alone, provide much information on the spectral
distribution. As an example, the $N\times N$ Jordan block $J_N$ of
eigenvalue zero and its perturbation on the lower-left corner
$J_{N,\eps}=J_N+\eps E_{N,1}$ both satisfy $\Tr J^n=0$ for
$n=1,\ldots, N-1$. However, their spectra are quite different:
$\Spec(J_{N,\eps})$ consists in $N$ equidistant points of modulus
$\eps^{1/N}$. Only the trace $\Tr J^N$ distinguishes these two
spectra.

\medskip

By the Stone-Weierstrass theorem, any function $f\in C^0(\SP^1)$ can
be uniformly approximated by  trigonometric polynomials. Hence, for
any $\epsilon$ there exists a Fourier cutoff $K\in\N$ such that the
truncated Fourier serie of $f$ satisfies 
\begin{equation}\label{trunc}
f^{(K)}(\theta)=\sum_{k=-K}^K f_ke^{2i\pi k\theta}\quad
\textrm{satisfies}\quad \| f_K-f\|_{L^{\infty}}\leq \epsilon\,. 
\end{equation}
We first study the average of $f^{(K)}$ over the angles
$(\theta_j)$:
$$
h\sum_{j=1}^{h^{-1} }f^{(K)}(\theta_j)=\sum_{k=-K}^K
f_k h\sum_{j=1}^{h^{-1}}e^{2i\pi k\theta_j}\,.
$$
For each power $k\in [-K,K]$, we relate as follows the sum over the
angles to the trace $\Tr\Mak^k$ : 
\begin{equation}\label{e:k-sum}
h\sum_{j=1}^{h^{-1}} e^{2i\pi k\theta_j} =\frac{h}{\moy
a^k}\Tr(\Mak^k)+ h\sum_{j=1}^{h^{-1}} \frac{\moy a^k -r_j^k}{\moy
a^k}\, e^{2i\pi k\theta_j}\,. 
\end{equation}
 Although this relation holds for
any nonsingular matrix $M$, it becomes useful when one notices that
most radii $r_j$ are close to $\moy a$: this fact allows to show
that the second term in the above right hand side is much smaller than the first
one.

Take $\delta>0$ arbitrary small, and denote $I_\delta\defeq[\moy
a-\delta,\moy a+\delta]$. From Theorem.~\ref{th:ergo}, $(i)$ we
learnt that
$h\#\{ r_j\in I_\delta\}=1+o_h(1)$. Hence the second
term in the right hand side of \eqref{e:k-sum} can be split into:
\begin{equation}\label{e:term2} 
h\sum_{j=1}^{h^{-1}}  \frac{\moy a^k -r_j^k}{\moy
a^k}\, e^{2i\pi k\theta_j} =h\sum_{ r_l\in I_\delta}\frac{\moy
a^k -r_l^k}{\moy a^k}\,e^{2i\pi k\theta_l}+ o_h(1)\,. 
\end{equation}
 By
straightforward algebra, there exists $C_K>0$ such that
$$
\forall r\in I_\delta,\ \forall k\in [-K,K],\qquad \frac{|\moy a^k
-r^k|}{\moy a ^k}\leq C_K\,\delta\,,
$$
so that the left hand side in \eqref{e:term2} is bounded from above by
$C_K\,\delta + o_h(1)$. By summing over the Fourier
indices $k$, we find
$$
h\sum_{j=1}^{h^{-1}} f^{(K)}(\theta_j)=
f_0+\sum_{\stackrel{k=-K}{k\neq 0}}^K  \frac{f_k}{\moy a
^k}\,h\,\Tr(\Mak^k)+\cO_{K,f}(\delta)+o_h(1)\,.
$$
Using the trace estimates of Proposition~\ref{p:traces} for the traces up to
$|k|=K$, we thus obtain
$$
h\sum_{j=1}^{h^{-1}} f^{(K)}(\theta_j) = f_0 +
\cO_{K,f}(\delta)+o_h(1)\,.
$$
Since this is true for every $\delta>0$, we deduce:
$$
h\sum_{j=1}^{h^{-1}} f^{(K)}(\theta_j)=f_{0} +
o_h(1)\,.
$$
We now  use the estimate \eqref{trunc} to write:
$$
h\sum_{j=1}^{h^{-1}} f(\theta_j)=f_0 +o_h(1) +
\cO(\epsilon)\,.
$$
$\epsilon$ being arbitrarily small, this concludes the proof.
\stopthm

\section{The Anosov case}\label{s:anosov}
\subsection{Width of the spectral distribution}
If $\ka$ is Anosov, one can obtain much
more precise spectral asymptotics using dynamical information about the \emph{decay of correlations} of classical observables under the dynamics generated by $\ka$. We will make use of probabilistic notations: a symbol
$a$ is seen as a random variable, and its value distribution will be
denoted $P_a$. If we denote as before the Lebesgue measure by $\mu$,
this distribution is defined for any interval $I\in\R$ by :
\begin{align*}
P_a(I)&\defeq \mu(a^{-1}(I))\\
&=\int_I P_a(dt)\,.
\end{align*}
This is equivalent to the following property: for any  continuous
function $f\in C(\R)$, one has
$$
\E(f(a))\defeq\int_{\t2}f(a)d\mu=\int_{\R}f(t)P_a(dt)\,.
$$
We now state a key result concerning $\lan$ when $\ka$ is Anosov.
\begin{lem}\label{l:decay}
Set $\la=\log\moy a$, and 
$$ x_n\defeq  \fr{1}{ n}\sum_{i=1}^n \log|a\circ\ka^i| -
\ell a =\lan-\ell a\,.
$$
If  $\ka$ is Anosov, we
have
$$
\limsup_{n\to\infty} \E(n x_n^2)<\infty\,.
$$
\end{lem}
\begin{proof}
Denote $f_i=\log|a\circ\ka^i|-\ell a $ and define the correlation function $c_{ij}$ as
$$
c_{ij}=\E(f_i f_j)\,.
$$
Then,
$$
\E(nx_n^2)=\fr{1}{n}\sum_{i=1}^n\sum_{j=1}^n c_{ij}\,.
$$
But for $\ka$ Anosov,  $|c_{ij}|\lesssim \e^{-\rho|i-j|}$ for some $\rho>0$ (see \cite{liv}). Hence 
$$\sum_{i=1}^n\sum_{j=1}^n c_{ij}=\cO(n)\,,$$
 and the proposition follows easily.
\end{proof}

\emph{Proof of \thmref{th:width}}. We now have all the tools to get our main result. In this paragraph, we will assume
that $n= n_\tau=E(\tau\log  h^{-1})$, with $\tau<T_{a,\ka}$  as above.  It will be more
convenient to show the following statement: for any $\eps>0$ and $C>0$,
\begin{equation}
 \lim_{h\rightarrow 0}\,h\,\#\Big\{1\leq
j\leq h^{-1} : \left|\log|\lambda_j^{(h)}|-\log\moy a\right|\leq C (
\log h^{-1})^{\eps-1/2}\Big\}= 1\,.
\end{equation}
For $h$ small enough, this equation is equivalent to \eqref{e:width} because $|\lambda_j^{(h)}|\geq a_-/2$. We will proceed exactly as in section \ref{s:radial}, but now $\delta$ and $\gamma$ will depend on $h$.
 
First we define as before two positive sequences $(\delta_h)_{h\in]0,1[}$ and
$(\gamma_h)_{h\in]0,1[}$ going to 0 as $h\to0$, and such that
$$
\fr{1}{\gamma_h\rc{\log h^{-1}}}\hto0 0 \quad\trm{and}\quad \fr{\gamma_h}{\delta_h}\hto0 0.
$$
 A simple choice can be made
by taking  $\delta_h\propto(\log h^{-1}) ^{\eps-\fr12}$ and $\gamma_h\propto(\log
h^{-1})^{\fr{\eps}{2}-\fr12}$, for some  $\eps\in]0,\fr12[$.

We call $\ell s_i$ the (positive) eigenvalues of $\lsn(a)$, and
define the integer $d_h$ such that
$$
d_h=\#\{1\leq i \leq h^{-1} : \log |\lambda_i| -\la\geq
\delta_h\}\,.
$$
The Weyl inequalities imply
\begin{equation}\label{eq:weyl0}
d_h(\la +\delta_h)\leq\sum_{i=1}^{d_h} \log
|\lambda_i|\leq \sum_{i=1}^{d_h}\ell s_i\,.
\end{equation}
Among the $d_h$ first (therefore, largest) numbers $(\ell
s_i-\la)_{i=1,\ldots,d_h}$, we now distinguish the $d_h'$ first
ones which are larger than $\gamma_h$, and
call
$$d_h-d_h'=\#\{1\leq i\leq d_h,\;
\ell s_i -\ell a< \gamma_h\}$$ the
number of remaining ones. Hence :
$$
d_h'=\#\{1\leq i \leq h^{-1} : \ell s_i -\la\geq\gamma_h\}\,.
$$
Substracting $d_h\ell a$ in Eq. \eqref{eq:weyl0} and noticing that
$d_h-d_h'\leq h^{-1} $, we get :

\begin{equation}\label{weyl2}
d_h\leq \fr{\gamma_h}{h\delta_h} + \fr{1}{\delta_h}\sum_{i=1}^{d_h'} (\ell s_i -  \ell a)\,.
\end{equation}

Recall that $\lsn -\ell a=\opn(\lan-\la)+\Ohh(h^{\alpha})$ for 
$\alpha=\sigma^->0$. From now on, $\alpha $ will denote a strictly positive
constant which value may change from equation to equation. Hence,
the sum in the right hand side of Eq. \eqref{weyl2} can be expressed
as :
$$
\sum_{i=1}^{d_h'} (\ell s_i -  \la)=\Tr
\Id_{[\gamma_h,2]}\lp\opn(\lan-\la)+\cO(h^{\alpha})\rp\,.
$$
The function $\Id_{[\gamma_h,2]}$ can easily be
smoothed to give a  function $\mathcal I_h$ such that $\mathcal
I_h(x)=0$ for $x\in\R\setminus[\gamma_h/2,3]$ and
$\mathcal I_h(t)=t $ on $[\gamma_h,2]$. Such a function
can be clearly chosen $w-$admissible with $ (\log
h^{-1})^{-1/2+\eps}\lesssim w(h)$. Note also that the logarithmic decay of $w$ in
this case will always  make the function $\mathcal I_h$ suitable
for the functional calculus expressed in Proposition \ref{p:calcfunct} and
Corollary \ref{c:calcfunctpert} -- see the discussion at the end of $\S$\ref{s:Sn}.
Continuing from these remarks, we obtain :
\begin{align*}
h\, d_h&\leq \fr{\gamma_h}{\delta_h}+\fr{1}{\delta_h}h\Tr \mathcal I_h\lp\opn(\lan -\la) +\Ohn(h^{\alpha})\rp\\
&\leq  \fr{\gamma_h}{\delta_h} + \fr{1}{\delta_h}h\Tr \opn(\mathcal I_h(\lan-\la)) +\cO(h^{\alpha}) \\
&\leq \fr{\gamma_h}{\delta_h} + \fr{1}{\delta_h} \int_\R\mathcal I_h(x)P_{x_n}(dx) +\cO(h^{\alpha})\,,
\end{align*}
where the functional calculus with perturbations  has been used. We now remark that $\forall
x\in \Supp \mathcal I_h$, one has
\begin{equation}\label{e:ineqtrace}
\mathcal I_h(x) \lesssim x^2\rc n\,.
\end{equation}
Indeed, we can clearly choose $\mathcal I_h$ such that $|\mathcal
I_h'(x)|\lesssim 1$. 
Hence $\mathcal I_h(x)\lesssim x
$, but since for $C>0$ fixed  we have $C(\log
h^{-1})^{-\fr12}\leq  x$ for $h$ small enough and $x\in \Supp \cI_h$, we get 
$$x\lesssim x^2\rc{\log
  h^{-1}}\,,
$$
which imply Eq. \eqref{e:ineqtrace}. Using Lemma
\ref{l:decay}, we continue from these remarks and obtain 
\begin{align*}
\int_\R \rc n\mathcal I_h(x)P_{x_n}(dx)&\leq \int_\R n
x^2 P_ {x_n}(dx)=\E (n x_n^2)<\infty\,,
\end{align*}
from which we conclude that
$$
h\, d_h\leq \fr{\gamma_h}{\delta_h} + \cO(\fr{1}{\rc{\log h^{-1}}\delta_h})
+\cO(h^{\alpha})\hto0 0\,.
$$
The theorem is completed by using the inverse map $\Mak^{-1}$,
exactly as for the proof of Theorem \ref{th:strip}, since $\|\opn(a^{-1}\circ\ka)U_h(a)^{-1}-\Mak^{-1}\|=\cO(h)$. \stopthm

\subsection{Estimations on the number of ``large'' eigenvalues.}

In this paragraph, we address the question of counting the eigenvalues
of $\Mak$ outside a circle with radius \emph{strictly} larger than
$\moy a$ when $\ka$ is Anosov.  As we already noticed, informations on
the eigenvalues of $\Sn(a)$ can be obtained from the function $a_n$, via the functional calculus. Hence, if we want to count the eigenvalues of $\Mak$ away from the average $\moy a$, we are lead to estimate the function $a_n$ away from its typical value $\moy a$. More dynamically speaking, we are interested in large deviations results for the map $\ka$. For $a\in \ctc$, these estimates take usually the following form \cite{op,rby}:
\begin{thm}[Large deviations] \label{th:dev} Let $c>0$ be a positive
  constant, and define 
$$\lc= \log (1+c/\moy a)\,.
$$
If $\ka$ is Anosov, there exists a function $I : \R_+\mapsto \R_+$,
positive, continuous and  monotonically increasing, such that 
\begin{equation}\label{e:dev}
\limsup_{n\to\infty}\fr{1}{n}\log \mu\lp x: x_n \in [\lc, +\infty[\rp \leq -I(\lc)\,.
\end{equation}
In particular, for  $n\geq 1$ large enough, one has
\begin{equation}\label{e:dev2}
P_{x_n}([\lc,\infty[)=\cO(\e^{-n I(\lc)})\,.
\end{equation}
\end{thm}
We now proceed to the proof of the theorem concerning the ``large'' eigenvalues of $\Mak$.

\emph{Proof of \thmref{th:large}.} 
As before, we prove the result for $\lsn (a)$, the extension to
$\Sn(a)$ being straightforward by taking the exponential. Define as
before 
$$d_h= \#\{1\leq j\leq h^{-1}:\log |\lambda_j^{(h)}|\geq \la +
\lc\}\,.$$ 
We also choose a small $\rho>0$, and define 
$$
d_h'=\#\{1\leq j\leq h^{-1}:\ell s_j\geq \la + \lc -\rho\}\,.
$$
The Weyl inequality can be written :
$$
d_h(\la+\lc)\leq \sum_{j=1}^{d_h'} \ell s_j + (d_h-d'_h)(\ell a+\lc-\rho)\,.
$$
Substracting $d_h(\la+\lc-\rho)$ on both sides, we get for $\rho$
small enough
$$
d_h\,\rho\leq \sum_{j=1}^{d_h'} (\ell s_j-\ell a)  - d_h'(\lc-\rho)\leq \sum_{j=1}^{d_h'} (\ell s_j-\ell a)\,.
$$
If we choose $n=n_\tau=E(\tau\log h^{-1})$ and $\tau<T_{a,\ka}$ as before, we can use exactly the same methods as above to evaluate  the right hand side of the preceding equation. Let $\mathcal I$ be a smooth function with $\mathcal I=1 $ on $[\lc-\rho,2]$, and $ \mathcal I=0$ on $\R\setminus [\lc-2\rho,3]$. We have 
\begin{equation*}
h\sum_{j=1}^{d_h'} (\ell s_j-\ell a)\leq h\Tr \mathcal I \lp \opn(\lan-\la)+\Ohn(h^{\sigma^-})\rp\,.
\end{equation*}
Recall that $\lan- \la\in S_{\delta^+}^0$ and $\sigma=1-2\delta-6\eta=1-\tau/T_{a,\ka}$. We easily check that $\sigma^-$ satisfies the condition $(ii)$ of Corollary \ref{c:calcfunctpert} since $\cI$ does not depend on $h$. 
Hence, 
$$
h\sum_{j=1}^{d_h'} (\ell s_j-\ell a)\leq \int_{\t2} \bbbone_{[\lc-2\rho, 3]}(\lan-\la) + \cO(h^{r})\,,\ \ r>0.
$$
Let us show that in fact, $r=\sigma^-$. By the functional
calculus, we have 
$$
r=\min\{1-2\delta^+-6\eps,\sigma^- -4\eps\}=\min\{1-2\delta^+-6\eps, 1-2\delta^+-6\eta^+ -4\eps\}
$$
where now, $\eps>0$ is arbitrary since $\mathcal I$ does not depend on $h$: this implies immediately  $r= \sigma^-$.
Using \eqref{e:dev2}, we get
$$
h \,d_h\leq \fr{1}{\rho}P_{x_n}([\lc-2\rho,\infty[)+\cO(h^{\sigma^-}) =\cO(\fr{1}{\rho}\e^{-n I(\lc-2\rho)}) +\cO(h^{\sigma^-})\,.
$$
Since $\rho$ is  arbitrarily small, we end up with
\begin{equation}\label{bound1}
h \,d_h =\cO(h^{\tau I(\lc^-)} +h^{\sigma^-})=\cO( h^{\tau I(\lc^-)} +h^{1-\tau/T_{a,\ka}^-})=\cO(h^{\min\{\tau I(\lc^-),1-\tau/T_{a,\ka}^-\}})\,.
\end{equation}
It is now straightforward to see that the bound is minimal if we choose
$$
\tau=\tau_c=\fr{T_{a,\ka}^-}{1+I(\lc^-)T_{a,\ka}^-}\,\,.
$$
\stopthm

\section{Numerical examples}\label{s:numeric}

In this section, we present a numerical illustration of Theorems
\ref{th:ergo} and \ref{th:width} for a simple example, the well known quantized cat
map. 

\subsection{The quantized cat map and their perturbations}
We represent a point of the torus $x\in\t2$ by a vector of $\R^2$ that we denote $X=(q,p)$.
Any matrix $A\in SL(2,\Z)$ induces an invertible symplectic flow $\ka : x\mapsto x'$ on  $\t2$ via a
transformation of the vector $X$ given by :
$$
X'=AX \mod 1\,.
$$
The inverse transformation is induced by $A^{-1}$, with $X=A^{-1}X'
\mod 1$. If $\Tr A>2$, this classical map, known as ``cat map'',  has  strong
chaotic features : in particular, it has the Anosov property (which implies
ergodicity). Any quantization $U_h(A)$ (see \cite{deg}) of $A\in SL(2,\Z)$ with $\Tr A>2$ satisfies
the hypotheses required for the unitary part of the maps
\eqref{Qmap}, and  the Egorov estimate \eqref{e:Egorov} turns out to
be exact, i.e. it holds without any remainder term.

Let us give a concrete example that will be treated numerically
below. Let $m\in \N$ and $A\in SL(2,\Z)$ of the form:
\begin{equation}\label{catmatrix}
A=\begin{pmatrix}
2m&1\\
4m^2-1&2m
\end{pmatrix}
\end{equation}
Then, in the position basis we have
\begin{equation}
U_h(A)_{jk}=\rc{h}\exp 2\i\pi h [mk^2-kj+mj^2]\,.
\end{equation}

We can also define some simple perturbations of the cat maps, by multiplying $U_h(A)$
with a matrix of the form $\exp (\fr{-2\i\pi}{h} \opn(H))$, with $H\in
\ctc$ a real function, playing the role of a ``Hamiltonian''. If we denote by $\e^{X_H}$ the  classical 
flow generated by $H$ for unit time, the total classical map will be  $\ka\defeq\e^{X_H}\circ\ka_A$.
For reasonable choices of ``Hamiltonians'' $H$, $\ka$ still define an Anosov maps on the torus, and 
the operators
\begin{equation}
\ti U_h(A,H)=e^{-\fr{2i\pi}{h}\opn( H)}U_h(A)
\end{equation}
quantize the map $\ka$. 
In our numerics, we have chosen $m=1$, and 
$$\opn(H)=\fr{\alpha}{4\pi^2} \sin(2\pi q)$$ 
with $\alpha=0.05$. This operator is diagonal in the position representation, and for $\alpha<0.33$ \cite{bk}, the classical map $e^{X_H}\circ\ka_A$ is Anosov. 

Remark that because of the perturbation, Egorov property \eqref{e:Egorov} now holds with some nonzero remainder term a priori. Such perturbed cat maps do not present in general the numerous spectral degeneracies caracteristic for the non-perturbed cat maps \cite{deg}, hence they can be seen as a classical map with generic, strong chaotic features.

\begin{figure}
\includegraphics[width=0.4\columnwidth, angle=0]{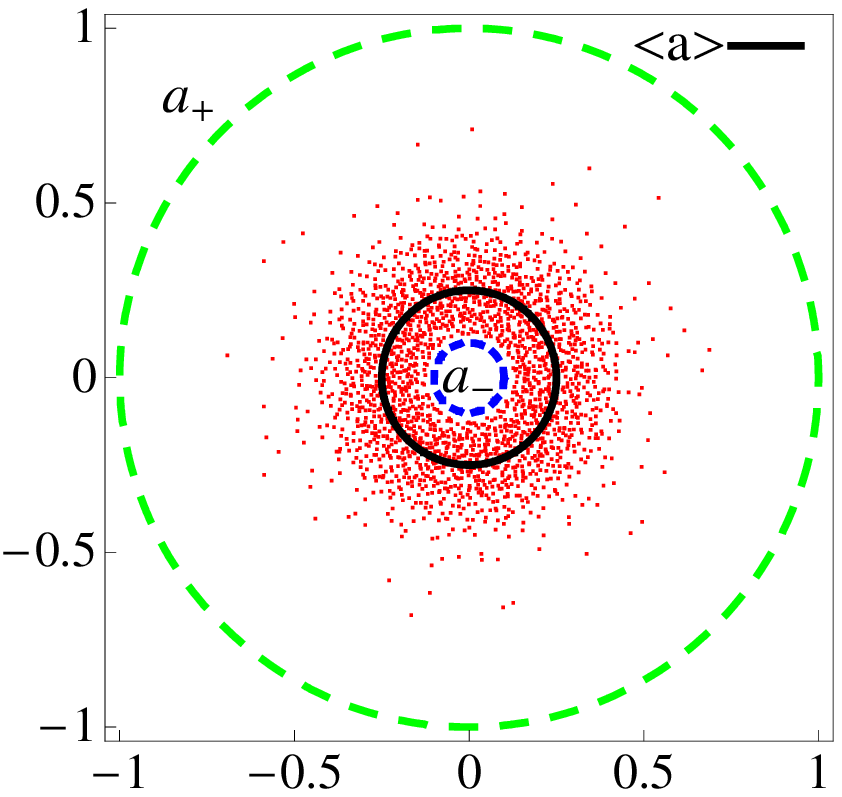}
\includegraphics[width=0.4\columnwidth, angle=0]{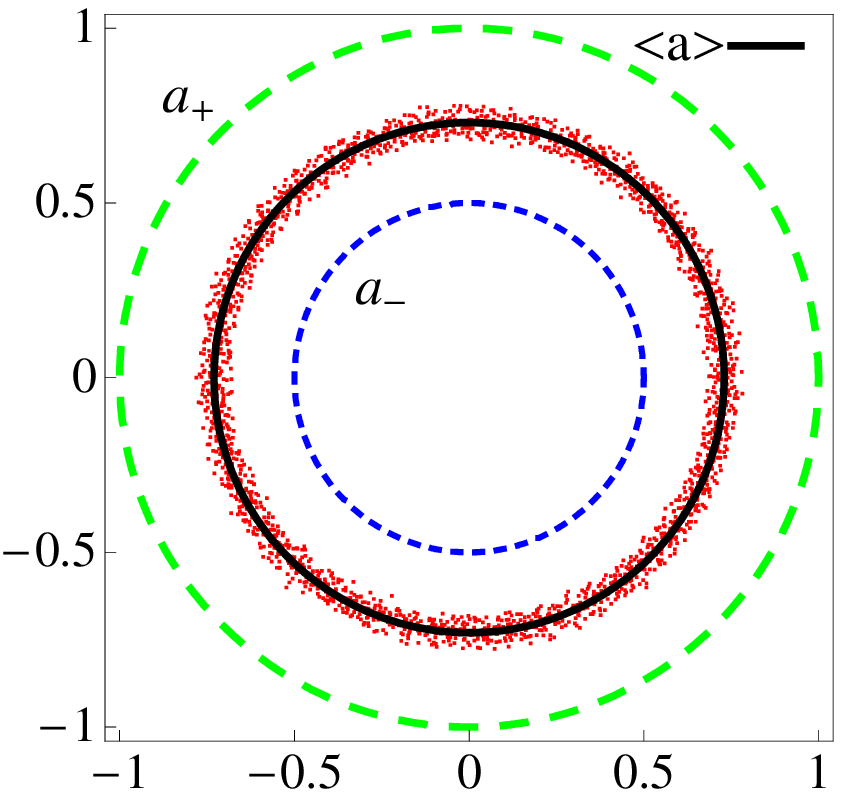}
\caption{Spectrum of $\Mak$ in the complex plane for $h^{-1}=2100$. The dashed circles correspond to $a_+$ and $a_-$, while the plain circle has radius $\moy a$.  To the left, we plot $a=a_1$, and  $a=a_2$ to the right.}
\label{spectres}
\end{figure}
For the damping terms, we choose two symbols of the form $a(q)$, whose quantizations are then diagonal matrices with entries $a(hj)$, $j=1,...,h^{-1}$. The function $a_1(q)$ has a plateau $a_1(q)=1$ for $q\in[1/3,2/3]$, another one for $q\in [0,1/6]\cup[5/6,1]$ and varies smoothly inbetween. For the second one, we take $a_2(q)=1-\fr12\sin(2\pi q)^2$. Numerically, we have computed
$$
\moy{a_1}\approx 0.250\,\quad\trm{and}\quad \moy {a_2}\approx0.728\,.
$$
Fig.  \ref{spectres} and \ref{densites} represent the spectrum of our perturbed cat map for $h^{-1}=2100$, with dampings $a_1$ and $a_2$. The spectrum stay inside an anulus delimited by $a_{i+}$ and $a_{i-}$, as stated in Eq. \eqref{e:strict}, and  the clustering of the eigenvalues around $\moy a_i$ is remarkable. For a more quantitative observation, the integrated radial and angular density of eigenvalues for
different values of $h^{-1}$  are represented in Fig. \ref{spectres}. We check that for moduli, the curve jumps around $\moy a$, which denotes a
maximal density around this value, and we clearly see the homogeneous angular  repartition of the eigenvalues of $\Mak$.

\begin{figure}
\includegraphics[width=0.4\columnwidth, angle=0]{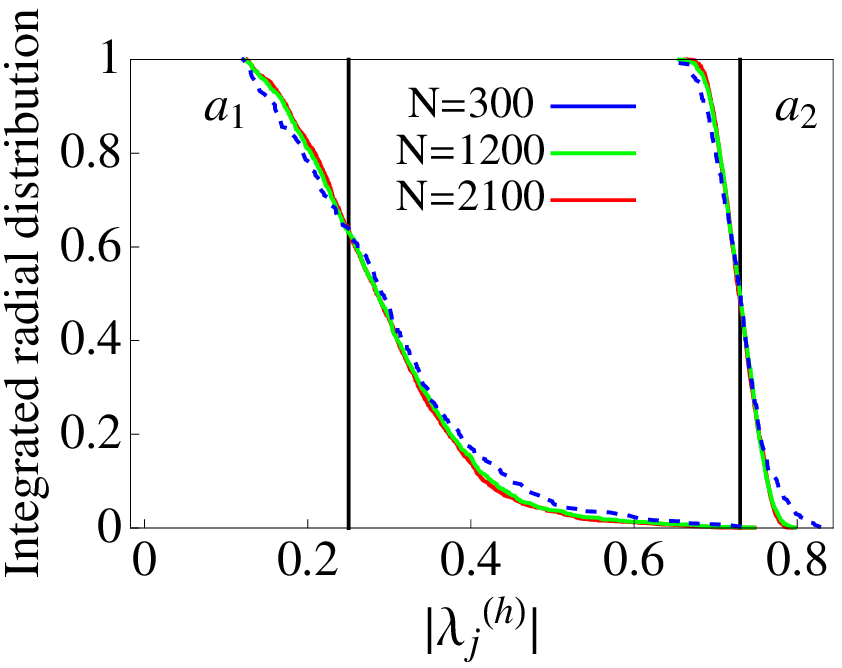}
\includegraphics[width=0.4\columnwidth]{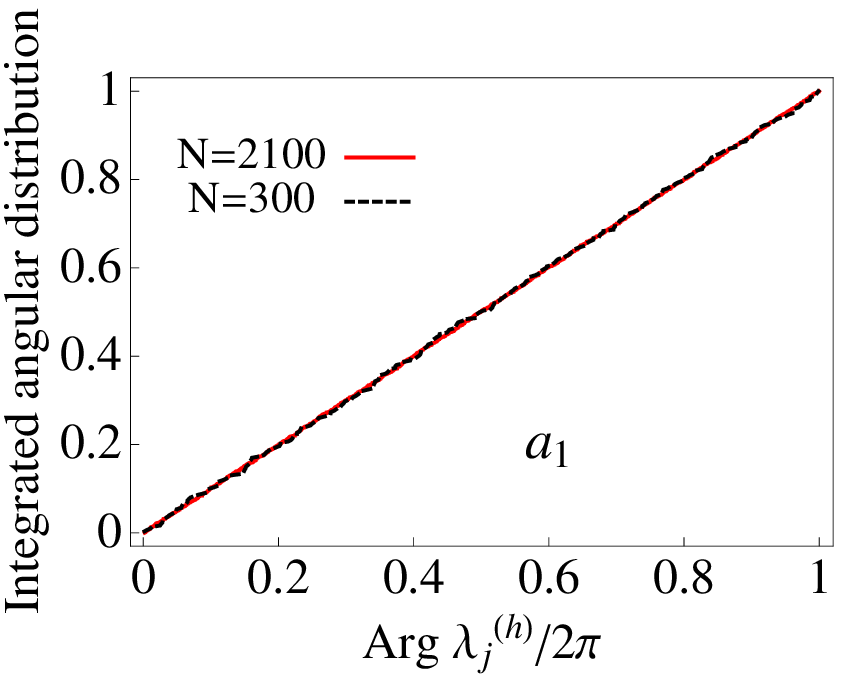}
\caption{Integrated spectral densities for the perturbed cat map. The radial distribution is represented to the left, the vertical bars indicate the value $\moy a$ for $a=a_1$ and $a=a_2$.  The angular distribution is represented to the right for the map $M_h(a_1,\ka)$.}
\label{densites}
\end{figure}

\begin{figure}
\includegraphics[width=0.5\columnwidth, angle=0]{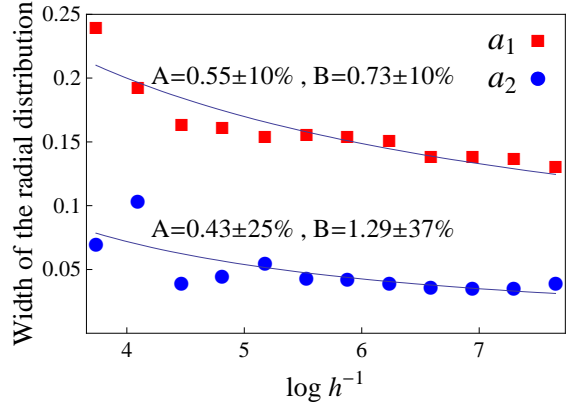}
\caption{Width of the radial distribution, together with the best 2-parameter fits $A (\log h^{-1})^{-B}$ and the asymptotic standard errors.}
\label{f:width}
\end{figure}

As we can observe in Fig. \ref{densites}, the width of the jumps do
not depend a lot upon $h$, at least for the numerical range we have
explored. This behavior could be explained by  \thmref{th:width},
which states that the speed of the clustering may be governed by $\log
h^{-1}$. To check this observation more in detail, we define the \emph{width}
$W_h$ of the spectral distribution of $\Mak$ as 
$$
W_h=|\lambda^{(h)}_{E(\fr{1}{4h})}|-|\lambda^{(h)}_{E(\fr{3}{4h})}|\,,
$$
and plot $W_h$ as a function of $h$ -- see Fig. \ref{f:width}. We clearly observe a decay with $h^{-1}$, although the 2-parameter fits $A(\log h^{-1})^{-B}$ hints a decay slightly faster than $(\log h^{-1})^{-1/2}$. Other numerical investigations presented in \cite{ns} for the quantum baker map show the same type of decay, and a solvable quantization of the baker map allow to compute explicitely the width $W_h$, which turns to be exactly proportional to $\rc{\log h^{-1}}$. This result, together with the numerics presented above,  seems to indicate that the bound on the decay of the eigenvalue distribution  expressed by \thmref{th:width} may be optimal.

\subsection*{Acknowledgment}
I would like to thank very sincerely Stéphane Nonnenmacher for suggesting
this problem, and above all for his generous
help and patience while introducing me to the subject. I also have
benefitted from his careful reading
of preliminary versions of this work. I would also like to
thank Frédéric Faure and Maciej Zworski for helpful and enlightening discussions.

\end{document}